\documentclass[10pt,conference]{IEEEtran}
\usepackage[caption=false]{subfig}
\usepackage{amsmath,amssymb,amsthm,mathrsfs,xspace,stmaryrd,float,graphicx,
cite,hyperref}
\usepackage[all]{xy}

\newtheorem{theorem}{Theorem}
\newtheorem{definition}{Definition}
\newtheorem{lemma}{Lemma}
\newtheorem{proposition}{Proposition}
\newtheorem{example}{Example}

\makeatletter
\def\@IEEEsectpunct{.\ \,}
\def\paragraph{\@startsection{paragraph}{4}{\z@}{1.5ex plus 1.5ex minus 0.5ex}%
{0ex}{\normalfont\normalsize\sffamily\bfseries}}
\makeatother

\newcommand {\scrA} {\ensuremath{\mathscr{A}}}
\newcommand {\scrB} {\ensuremath{\mathscr{B}}}
\newcommand {\scrS} {\ensuremath{\mathscr{S}}}
\newcommand {\scrT} {\ensuremath{\mathscr{T}}}

\newcommand {\calAcom}{\mathcal{A}\mathit{com}} 
\newcommand {\calLI}{\mathcal{L}\mathbf{1}}

\newcommand {\minc}{\mathit{min}}
\newcommand {\maxc}{\mathit{max}}

\newcommand \mso {\ensuremath{\mathrm{MSO}}\xspace}
\newcommand \fo {\ensuremath{\mathrm{FO}}\xspace}

\newcommand \bet {\ensuremath \mathit{bet}}

\newcommand {\N}{\mathit{N}}

\newcommand \bfS {\ensuremath{\mathbf{S}}}

\newcommand \bfV {\ensuremath{\mathbf{V}}}
\newcommand \calV {\ensuremath{\mathcal{V}}}

\newcommand \calC {\ensuremath{\mathcal{C}}}

\newcommand \bfD {\ensuremath{\mathbf{D}}}

\newcommand \Reg {\mathit{Reg}}

\newcommand \bfW {\ensuremath{\mathbf{W}}}
\newcommand \calW {\ensuremath{\mathcal{W}}}

\newcommand \bbN {\mathbb{N}}
\newcommand \supp {\mathit{Supp}}

\newcommand {\ssim} {\approx}
\renewcommand {\S} {\mathscr{S}}

\newcommand {\da} {\dagger}

\newcommand \Acom {\mathbf{Acom}}
\newcommand \LI {\mathbb{L}\mathbf{1}}

\newcommand \sdp {\mathbin{**}}

\newcommand {\pref} {\mathit{pref}}
\newcommand {\suff} {\mathit{suff}}
\newcommand {\id} {\mathit{id}}
\newcommand {\flip} {\mathit{flip}}

\makeatletter
\newcommand*{\@old@slash}{}\let\@old@slash\slash
\def\slash{\relax\ifmmode\delimiter"502F30E\mathopen{}\else\@old@slash\fi}
\makeatother

\renewcommand {\(}{\left(}
\renewcommand {\)}{\right)}
\newcommand \ltt {\ensuremath{\mathrm{LTT}}\xspace}

\newcommand \lrtt {\ensuremath{\mathrm{wLRTT}}\xspace}

\newcommand \callrtt {\ensuremath{\mathit{w}\mathcal{LRTT}}\xspace}

\newcommand \op {\ensuremath{\mathit{op}}}

\newcommand{\approxr}[2]{\stackrel{r}{\approx}\mathrel{{}^{#1}_{#2}}}

\makeatletter
\thm@headfont{
  \ensuremath{\blacktriangleright}\nobreakspace\sffamily\bfseries}

\theoremstyle{plain}
\newtheorem{thm}{\protect\theoremname}
\theoremstyle{definition}

\theoremstyle{plain}

\theoremstyle{plain}

\theoremstyle{definition}

\newtheorem{example}[thm]{\protect\examplename}

\usepackage{tikz}
\usetikzlibrary{positioning}
\usepackage{xcolor}
\usepackage[all]{xy}

\makeatother

\providecommand{\corollaryname}{Corollary}
\providecommand{\definitionname}{Definition}
\providecommand{\examplename}{Example}
\providecommand{\lemmaname}{Lemma}
\providecommand{\theoremname}{Theorem}

\begin{document}

\title{An Algebraic Characterisation of First-Order Logic with Neighbour}

\author{\IEEEauthorblockN{Amaldev Manuel}
\IEEEauthorblockA{Indian Institute of Technology Goa\\
amal@iitgoa.ac.in}
\and
\IEEEauthorblockN{Dhruv Nevatia}
\IEEEauthorblockA{Chennai Mathematical Institute\\
dhruv@cmi.ac.in}
}

\IEEEoverridecommandlockouts
\IEEEpubid{\makebox[\columnwidth]{978-1-6654-4895-6/21/\$31.00~
		\copyright2021 IEEE \hfill} \hspace{\columnsep}\makebox[\columnwidth]{ }}
	
\maketitle
\begin{abstract}
We give an algebraic characterisation of first-order logic with the neighbour relation, on finite words. For this, we consider languages of finite words over alphabets with an involution on them. The natural algebras for such languages are involution semigroups. To characterise the logic, we  define a special kind of semidirect product of involution semigroups, called the locally hermitian product. The characterisation theorem for FO with neighbour states that a language is definable in the logic if and only if it is recognised by a locally hermitian product of an aperiodic commutative involution semigroup, and a locally trivial involution semigroup.
We then define the notion of involution varieties of languages, namely classes of languages closed under Boolean operations, quotients, involution, and inverse images of involutory morphisms. An Eilenberg-type correspondence is established between involution varieties of languages and pseudovarieties of involution semigroups.   
\end{abstract}. 


\section{Introduction}
We give an algebraic characterisation of a logic over finite words, namely \emph{the first-order logic with the neighbour relation}. Let $A$ be a finite alphabet. Formulas of the logic FO with neighbour, $\fo(A, \minc, \maxc, \N)$, are interpreted over finite words over the alphabet $A$. The constants $\minc$ and $\maxc$ denote the first and last positions of a given word respectively. The atomic formulas of the logic are the following: The predicate $P_{a}(x)$, for $a \in A$, denotes that the position $x$ is labelled by the letter $a$. The binary predicate $\N(x,y)$ denotes that $x$ and $y$ are neighbours, i.e., either $x+1=y$ or $y+1=x$. Finally, we have the equality predicate $x=y$.  The  set of formulas of the logic are closed under Boolean operations and first order quantifications, i.e., $\varphi\vee \psi,\varphi\wedge\psi, \neg \varphi,  \forall x\, \varphi$, and $\exists x\, \varphi$ are also formulas of the logic, if $\varphi$ and $\psi$ are formulas of the logic. For example, the formula 
\[
P_{a}(\minc) \wedge P_{b}(\maxc) \wedge \forall x \forall y\, \(\N(x,y) \rightarrow \(P_{a}(x) \leftrightarrow  P_{b}(y)\)\)
\] 
defines the language $a(ba)^{\star}b$ over the alphabet $\{a,b\}$. The language \emph{defined} by a formula $\varphi$, denoted as $L(\varphi)$, is the set of all words satisfying the formula $\varphi$.

Before we go into the characterisation problem of $\fo(A, \minc, \maxc, \N)$, it is interesting to note the expressibility of related logics. The neighbour and {\em between} predicates --- 
the ternary predicate $\bet(x,y,z)$ is true if position $y$ is strictly between positions $x$ and $z$ --- were first studied in \cite{PAG,PAG-Arxiv}. Although
one could express that a position is one of the endpoints using $\bet$ as well as $\N$, the constants $\min$ and $\max$ themselves are not expressible using $\bet$ and $\N$.
Since both $\bet$ and $\N$ are their own left-to-right dual, if the constants $\min$ and $\max$ are not used then languages definable using these predicates are closed under the reverse operation.  For the sake of good algebraic properties, we include the constants $\min$ and $\max$ in our vocabulary. Next we state the results of \cite{PAG,PAG-Arxiv} in this setting.

The  monadic second-order logic $\mso(A, \minc, \maxc, \N)$, is the extension that also allows monadic second-order quantification --- $\exists X \varphi$, $\forall X \varphi$ are also formulas whenever $\varphi$ is a formula of the logic; The formula $\exists X \varphi$ is true if there is a set of positions $X$ that satisfies the formula $\varphi$. For example, the formula
\[
\exists X \forall x \forall y \(\N(x,y) \rightarrow \(X(x) \leftrightarrow  \neg X(y)\)\) \wedge X(\minc) \wedge \neg X(\maxc)  
\] 
defines all the words of even length. It turns out that $\mso(A, \minc, \maxc, \N)$ defines precisely all regular languages, i.e., it has the same expressive power as $\mso$ with the \emph{successor} relation, $\mso(A, +1)$, by B\"uchi-Elgot-Trakhtenbrot's theorem. The relationship between the predicates $\N$ and $\bet$ are analogous to their oriented
counter parts, namely, the the successor relation ($x+1=y$) and the order relation ($x<y$). Both $\N$ and $\bet$ are definable in terms of the other by monadic second-order formulas; thus $\mso(A, \minc, \maxc, \bet)$ also defines all regular languages. Likewise, $\fo(A, \minc, \maxc, \bet)$ defines precisely all aperiodic regular languages --- languages that are definable in the logic $\fo(A, <)$.
In fact, this parallelism between $\fo(A, \minc, \maxc, \bet)$ and $\fo(A, <)$ extends to their quantifier alternation hierarchies.

However, the parallel described  so far breaks down in the case of $\fo(A, \minc, \maxc, \N)$. There are languages expressible in $\fo(A, +1)$ that are not expressible in the former logic. For instance, the language 
\[
L = c^* abc^*
\] 
over the alphabet $\{a,b,c\}$ is expressible in the logic $\fo(A, +1)$. But, using an Ehrenfeucht-Fra\"iss\'e argument \cite{ebbinghaus2005finite, PAG}, it can be shown that $L$ is \emph{not} definable in $\fo(A, \minc, \maxc, \N)$. Thus, we have the question of characterising the languages definable in this logic.

Using Hanf's theorem \cite{ebbinghaus2005finite} from finite model theory, it is possible to give language-theoretic characterisations of both $\fo(A, +1)$ and $\fo(A, \minc, \maxc, \N)$.
For $t>0$, we define the \emph{equality with threshold $t$} on the set $\mathbb{N}$ of 
natural numbers by 
\begin{equation}
i=^{t}j ~:=~ \begin{cases} i=j & \text{if $i<t$},\\
                          i\geq t \text{ and } j\geq t & \text{otherwise}.
                          \end{cases}
                          \end{equation}
The word $y \in A^{+}$ is a \emph{factor} of the word $u \in A^{+}$ if $u = x y z$ for some $x,z$ in $A^*$.  We use $\sharp(u, y)$ to denote the number of times the factor $y$
appears in $u$, i.e.\ the number of pairs 
$(x,z)$, where $x,z \in A^{*}$, such that $u=xyz$.

\begin{definition}
\label{defn:approx}
Let $\approx_k^t$, for $k, t>0$, be the equivalence on $A^*$, whereby two words
$u$ and $v$ are equivalent if either they both have length at most $k-1$ and
$u=v$, or otherwise they have
\begin{enumerate}
\item the same prefix of length $k-1$,
\item the same suffix of length $k-1$,
\item and the same number of occurrences, up to threshold $t$, for all factors of length $ \leq k$, i.e.\ $\sharp(u, y)  =^{t} \sharp(v, y)$
for each word $y \in A^+$ of length at most  $k$.
\end{enumerate}
\end{definition}

A language is \emph{locally threshold testable} (or \ltt for short) if it is a
union of $\approx_{k}^{t}$ classes, for some $k,t > 0$.
Locally threshold testable languages are precisely the class of languages
definable in $\fo(A, +1)$ \cite{BeauquierPin,WolfgangThomas}.

 Since the neighbour predicate $\N$ is definable using the successor relation in first-order logic, $\fo(A,\minc, \maxc, \N)$ definable languages are a subset of \ltt.  But this inclusion is strict, as we have seen.

 We define a coarser equivalence $\approxr{t}{k}$ by the counting factors only up to reverse.
Let $\sharp^{r}(w,v)$ denote the 
number of occurrences of $v$ or $v^{r}$ in $w$, i.e.\ the number of pairs 
$(x,y)$, where $x,y \in A^{*}$, such that $w=xvy$ or $w=xv^{r}y$. 

\begin{definition}
\label{defn:approxr}
Let $k,t>0$.  Two words $w,w'\in A^{*}$ are $\approxr{t}{k}$-equivalent if $|w|<k$ and $w = w'$, or $w,w'$ both are of length at least $k$, and they have 
\begin{enumerate}
  \item the same prefix of length $k-1$,
  \item the same suffix of length $k-1$, and
  \item $\sharp^{r}(w,v) =^{t} \sharp^{r}(w',v)$ for each word $v \in A^{+}$ of length at most $k$.
\end{enumerate}
\end{definition}


It is shown in \cite{PAG} that a language is definable in $\fo(A, \N)$ if it is closed under the reverse operation and is a union of equivalence classes of  $\approxr{t}{k}$ for some
$k,t>0$. Such languages are called {\em locally-reversible threshold testable} languages. 
The class of languages accepted by $\fo(A, \minc, \maxc, \N)$ are characterised in a similar manner. A language $L$ is {\em weakly locally-reversible threshold testable}, \lrtt for
	short, if it is a union of equivalence classes of $\approxr{t}{k}$ for some
	$k,t>0$. By a straight-forward adaptation of the proofs in \cite{PAG,PAG-Arxiv} we can show that

\begin{theorem}
  A language is definable in the logic $\fo(A, \min, \max, \N)$ if and only if it is weakly locally-reversible
  threshold testable.
\label{Thm:fon}
\end{theorem}

Obtaining an algebraic characterisation for the class $\lrtt$ is an interesting problem.  The analogous characterisation of $\fo(A, +1)$  is a deep result \cite{BeauquierPin, Zalcstein1972, BrozozoskiSimon, McNaughton} in the theory of finite semigroups. The first observation is that $\ltt$ is a \emph{variety of languages} --- 
a class of languages that is closed under Boolean operations, quotients with respect to words, and inverse images under homomorphisms. Eilenberg's variety theorem states that varieties of languages correspond to pseudovarieties of semigroups ---  a pseudovariety of finite semigroups is a set of finite semigroups that is closed under finite direct products, subsemigroups and quotients. The characterisation of $\ltt$ proceeds by first writing the class as a semidirect product of pseudovarieties of aperiodic commutative monoids ($\Acom$) and righty trivial semigroups ($\bfD$ is the pseudovariety of semigroups that satisfy the equation $se=e$ for all elements $s$ and idempotents $e$). This step is the algebraic analogue of synthesising an automaton for  an $\ltt$ language as a cascade of a scanner and an acceptor \cite{BeauquierPin}.
The next step in the proof uses the framework of categories for obtaining the identities  capturing the semidirect product \cite{StraubingVD, Tilson, StraubingBook}.

In the case of $\lrtt$, this approach does not work; because \lrtt does not form a {variety of languages} ---  for instance, the language $(xy)^{*}ab(xy)^{*}$ over the alphabet $\{a,b,x,y\}$ is in \lrtt. Its inverse image under the morphism $a\mapsto a, b \mapsto b, c \mapsto xy$
is the language $c^{*}abc^{*}$, over the alphabet $\{a,b,c\}$. As we have seen, this language is not in $\lrtt$. Therefore, $\lrtt$ is not closed under inverse image of morphisms, and hence not a variety of languages. This necessitates a rubric, of the like of varieties and operations on varieties, to study the class $\lrtt$.

\paragraph*{Our results}
It was already observed in \cite{PAG} that one needs to extend semigroups with an involution (also called $\star$-semigroups) --- an involution on a semigroup $S$ is an operation $\star$ such that $(a^{\star})^{\star}=a$ and $(ab)^{\star} = b^{\star}a^{\star}$, for each $a,b \in S$ --- to characterise \lrtt. The involution operation is a generalisation of the reversal operation on words. An involutory alphabet $\scrA$ is a finite alphabet $A$ with a bijection $\da$ on it. The map $\da$ extends to words over $A$ as $(a_{1}\ldots a_{n})^{\da} = a_{n}^{\da}\ldots a_{1}^{\da}$, where each $a_{i} \in A$. This map is an involution on the semigroup $A^{+}$. We lift the notions of recognisability and syntactic semigroups to the case of languages over involutory alphabets. Syntactic algebras for such languages are semigroups with an involution. A finite alphabet $A$ can be seen as an involutory alphabet with the identity function $\id$ as $\da$; such alphabets are called \emph{hermitian}. This just means that $A^{+}$ is equipped with the reverse operation as involution. Syntactic algebras for languages over hermitian alphabets are semigroups with involutions that are generated by their hermitian elements ($a$ is hermitian if  $a^\star = a$).

Assume $(A, \da)$ and $(B, \star)$ are involutory alphabets and $h\colon A^{+} \rightarrow B^{+}$ is a morphism. The morphism $h$ is involutory if $h(w^{\da})=h(w)^{\star}$.
An \emph{involution variety of languages} maps each involutory alphabet $\scrA = (A, \da)$ to a class of languages over $A$ that is closed under involutions, Boolean operations, quotients with respect to words and inverse images under involutory morphisms. We establish an Eilenberg-type correspondence between involution varieties of languages and pseudovarieties of finite involution semigroups --- classes of finite involution semigroups that are closed under finite direct products, sub-$\star$-semigroups and quotients. 

If $\scrS=(S, \star)$ and $\scrT=(T, \da)$ are involution semigroups, a two-sided action of $T$ on $S$ is \emph{compatible} with the involution if $(tst')^\star = t'^\da s^\star t^\da$, where $tst'$ denotes the left action by $t\in T$ and right action by $t' \in T$ on the element $s\in S$. A compatible two-sided action of $T$ on $S$ is {\em locally hermitian} if for all idempotents $e\in T$ and elements $s\in S$ it holds that $ese^\da = es^\star e^\da$. Intuitively this means that within an
idempotent context that looks the same on the outside from either direction, one can reverse factors. Bilateral semidirect products with locally hermitian actions are called locally hermitian semidirect products. If $\bfV$ and $\bfW$ are pseudovarieties of involution semigroups, the locally hermitian  product of $\bfV$ and $\bfW$ is the pseudovariety generated by all locally hermitian semidirect products of involution semigroups from $\bfV$ and $\bfW$.
It is shown that a language is \lrtt if and only if it is recognised by a locally hermitian semidirect product of an aperiodic commutative semigroup and a locally trivial semigroup. This result can also be stated in terms of the corresponding pseudovarieties of involution semigroups $\Acom^\star$ and $\LI^\star$. Therefore, we have an algebraic characterisation of the class $\fo(A, \minc, \maxc, \N)$.

\paragraph*{Related work} The $\fo$ and $\mso$ logics using the predicates 
$\N$ and $\bet$ were introduced in \cite{PAG,PAG-Arxiv}. Formulas of these logics that do not use the constants $\minc$ and $\maxc$, are self-dual, i.e., replacing each predicate by its dual, for instance $N(x,y)$ by $N(y,x)$, results in an equivalent formula. 
Therefore, $\mso(A, \N)$ and $\mso(A, \bet)$  recognise precisely the class of \emph{reversible} regular languages --- a language is reversible if it is invariant under taking the reverse of words in it. Similarly, $\fo(A, \bet)$ recognises all aperiodic reversible languages. In \cite{PAG}, the question of deciding membership in these classes was studied. For the logics $\mso(A, \bet),\fo(A, \bet)$, the existing decidable characterisations of $\mso(A, <)$ and $\fo(A, <)$ immediately yield an answer. For the case of $\fo(A, \N)$, a partial answer was given; it was shown that if a language $L$ is definable in $\fo(A, \N)$, then its syntactic semigroup  satisfies the identity $ex^\star e^\star = ex e^\star$, in addition to those defining the class $\ltt$.

A different but related {\em between} predicate
(namely $a(x,y)$, for $a \in A$, is true if there is an $a$-labelled position
between positions $x$ and $y$) was introduced in
\cite{Straubing1,Straubing2,Straubing3}. They 
characterise the expressive power of two-variable first-order logic with the order relation  ($\fo^{2}(<)$)
enriched with the between predicates $a(x,y)$ for $a\in A$, and show an
algebraic characterisation of the resulting family of languages.

Since the class of involution semigroups contains many important classes of semigroups such as inverse semigroups, there is substantial literature on varieties of involution semigroups from a semigroup theoretic perspective (see \cite{Dolinka} for a survey). However we are not aware of a language theoretic study of pseudovarieties of finite involution semigroups or results similar to ours.

\paragraph*{Orgranisation of the paper} In Section \ref{Sec:2}, we describe involution semigroups and extend the notions of recognisability by semigroups to the case of languages with an involution. In  Section \ref{Sec:4}, involutory semidirect products are defined, and we introduce the notion of locally hermitian semidirect product. Then, it shown that a language is $\lrtt$ if and only if it is recognised by the locally hermitian semidirect product of an aperiodic commutative semigroup and a locally trivial semigroup. We introduce the notion of involution varieties of languages and pseudovarieties of involution semigroups, and prove an Eilenberg correspondence between them in Section \ref{Sec:3}.  In Section \ref{Sec:5}, we conclude with some avenues for further inquiry.


\section{Languages over Involutory Alphabets}
\label{Sec:2}

\subsection{Recognisable languages}

Let $A$ be a finite alphabet. Then, $A^+$ denotes the set of  nonempty finite words over $A$, and $A^*$ denotes the set of all finite words over $A$ including the empty word $\epsilon$. Extending this notation, $A^{\leq k}$ denotes the subset of $A^{*}$ consisting of words of length at most $k$, including $\epsilon$.

A semigroup $S$ is a set together with an associative binary operation. Unless specified, the semigroup operation is written as $x \cdot y$, for semigroup elements $x,y$; but we omit notating it whenever possible. If the operation has an identity, then $S$ is called a monoid. All semigroups and monoids we consider are finite, except when stated otherwise.
For a finite alphabet $A$, the set $A^+$ forms a semigroup under concatenation, while the set $A^*$ is a  monoid with the empty word $\epsilon$ as the identity. 
The set $S' \subseteq S$ forms a \emph{subsemigroup} of $S$ if $S'$ is closed under the semigroup operation. 

 Let $(T,+)$ be a semigroup. A \emph{semigroup morphism} from $S$ to $T$ is a map $h\colon S \rightarrow T$ that satisfies the equation $h(xy)=h(x)+h(y)$, for all $x,y \in S$. 
If $S$ and $T$ are monoids, then $h$ is a \emph{monoid morphism} if additionally $h$ maps the identity of $S$ to the identity of $T$. A semigroup $T$ {\em divides} the semigroup $S$ if $T$ is the  image of a subsemigroup of $S$ under some morphism.

Let $A$ be a finite alphabet. A language $L \subseteq A^+$ is {\em recognized} by a semigroup $S$ if there is a morphism $h\colon A^+ \rightarrow S$ and a subset $P \subseteq S$ such that $L = h^{-1}(P)$. The set $P$ is called the accepting set of $L$. The class of languages recognised by finite semigroups coincide with the class of regular languages \cite{holcombe}, p.~160.

The above definitions are easily adapted to the case of languages over $A^{*}$. The definitions are assumed in the subsequent sections.

\subsection{Languages over Involutory Alphabets}

A finite {\em involutory alphabet} $\scrA=(A, \da)$ is a finite alphabet $A$ with a bijection $\da\colon A \rightarrow A$  such that $(a^\da)^\da = a$ for each letter $a \in A$. The function $\da$ is extended to words over $A$ as
\begin{equation}
\label{eqn:defn-da}
w^\da = a_n^\da \ldots a_1^\da~,
\end{equation}
if $w=a_{1}\cdots a_{n} \in A^{+}$.  Clearly, ${w^{\da}}^{\da}=w$. Moreover, $(uv)^{\da}=v^{\da}u^{\da}$, for all words $u,v \in A^{+}$. Hence, the operation $\da \colon A^{+}\rightarrow A^{+}$ is an \emph{involution}: a function that is its own inverse. If $L \subseteq A^{+}$ is a language, then $L^{\da}$ denotes the set $\{ w^\da \mid w \in L\}$.

If $\da$ is the identity function on $A$, then $\scrA$ is called a \emph{hermitian} alphabet. In that case, $\da$ is the reverse operation $u \mapsto u^{r}$, for $u\in A^{+}$, and $L^\da$ is the set $L^{r}$, for $L\subseteq A^{+}$. Unless specified, a finite alphabet is taken to be a hermitian alphabet, i.e. with the reverse operation as the involution.

\begin{definition} An involution semigroup (also called a $\star$-semigroup) $\scrS=(S, \star)$ is a semigroup $S$ extended with an operation ${\star} \colon S \rightarrow S$ such that for all elements $a,b$ of $S$,
\[ 
  \(a^\star\)^\star = a~, \text{ and } \(a \cdot b\)^\star = b^\star \cdot a^{\star}~.
\]
\end{definition}

We say that $S$ is the semigroup reduct of $\scrS$. The semigroup $A^{+}$ with the operation $\da \colon A^{+} \rightarrow A^{+}$, defined in (\ref{eqn:defn-da}), forms an involution semigroup, denoted as $\scrA^+$. Similarly,  $A^{*}$ with  $\da$ forms a monoid with an involution, denoted as $\scrA^{*}$. The involution semigroup $\scrA^+$ is \emph{not} a \emph{free} involution semigroup in general; it is only so when the relation $\da$ is irreflexive on $A$ (\cite{freesem}, p.~172).

\begin{example}
\label{example:inv}
Let $T = \{a, b, ab, ba\}$ be the semigroup with the multiplication 
\[
aa = a~,\quad bb=b~,\quad aba=bab=ba.
\]
Let $\star$ be the map $a^\star = b$. The map extends to an involution on $T$, 
\[
b^\star = (a^\star)^\star = a~,\quad (ab)^\star = b^\star a^\star = ab~,\quad (ba)^\star=a^{\star}b^{\star}=ba~.
\]
However, the map $a^\da = a, b^\da=b$ is not a well-defined involution on $T$~; 
\[
(ba)^\da =a^{\da}b^{\da}=ab  \neq ba = aba = a^{\da}b^{\da}a^{\da} = (aba)^{\da} = (ba)^{\da}~.
\]
\end{example}

It is not possible to associate an involution with every semigroup; for instance, the semigroup of right zeros, where the multiplication follows the identity $xy =y$ for any elements $x,y$, does not admit an involution if it has more than one element. However, it is possible to obtain a $\star$-semigroup, from a given semigroup, that shares similar properties, as we shall see later.

Let $\scrS=(S, \star)$ and $\scrT = (T, \dagger)$ be involution semigroups. A morphism of involution semigroups,  also called an involutory morphism, from $\scrS$ to $\scrT$, denoted as $h\colon\scrS \rightarrow \scrT$, is a semigroup morphism $h \colon S \rightarrow T$ such that $h(a^\star)=h(a)^\dagger$, for all $a \in S$.

\begin{definition} Let $\scrA=(A, \da)$ be an involutory alphabet.
The involution semigroup $\scrS=(S, \star)$ \emph{recognises} the language $L \subseteq A^{+}$ if there is a morphism of involution semigroups $h\colon\scrA^{+}\rightarrow \scrS$ and a subset $P\subseteq S$ such that $L=h^{-1}(P)$.  
\end{definition}

 \begin{example}
Let $\scrA$ be the involutory alphabet $\{a,b\}$ with the map $a^\star = b$. 
Let $T = \{a, b, ab, ba\}$ be the semigroup from Example \ref{example:inv}. Then, 
$T$ accepts the language $L=a^+b^+$ with the morphism $h(a)=a, h(b)=b$ and the accepting set $\{ab\}$. 
\end{example}

Clearly, if a language is recognised by an involution semigroup, then it is recognised by a semigroup as well. The other direction also holds, as we show next.

Let $\scrS= (S, \da)$ be an involution semigroup and let $T$ be a semigroup. 
By $T^{\op}$, we denote the opposite of $T$ --- the semigroup with the same set of elements but reversed multiplication, i.e., $a\cdot^{\op} b = b \cdot a$, for all $a,b\in T$. Assume there is a semigroup morphism from $S$ to $T$. Let $h^{\da}$ be the map from $S$ to $T^{\op}$, given by, for each $s\in S$,
\[
h^{\da}(s) = h(s^{\da})~.
\]
\begin{lemma}
\label{lemma:h-star}
If $h\colon S\rightarrow T$ is a morphism, then $h^{\da}\colon S \rightarrow T^{\op}$ is also a morphism.
\end{lemma}
\begin{proof}
	For $x,y \in S$,
	\begin{align*}
		h^{\da}(xy)&=h((xy)^{\da})\\
		&=h(y^{\da} x^{\da})\\
		&=h(y^{\da})h(x^{\da}) && \text{(product in $T$)}\\
		&=h^{\da}(x)h^{\da}(y) && \text{(product in $T^{\op}$)}~.
	\end{align*}              
	Hence $h^{\da}\colon S \rightarrow T^{\op}$ is a semigroup morphism. 
\end{proof}

Let $\scrA = (A, \da)$ be an involutory alphabet. Assume that $L \subseteq A^{+}$ is recognised by the semigroup $S$ with the morphism $h$ and the accepting set $P$. Lemma \ref{lemma:h-star} gives that $h^{\da}\colon A^{+}\rightarrow S^{\op}$ is also a morphism; by definition, $h^{\da}(w^{\da}) = h(w)$, for all $w\in A^{+}$. Hence, $h^{\da}(L^{\da})=h(L)=P$; hence, the semigroup $S^{\op}$ recognises $L^{\da}$ with the morphism $h^{\da}$.

Next, we show that $\(T \times T^{op}, \flip\colon (x,y) \mapsto (y,x)\)$ is an involution semigroup. It suffices to show 
\begin{lemma}
\label{lemma:s-s-op}
 $\flip$ is an involution on $T \times T^{\op}$.
\end{lemma}
\begin{proof}
	Clearly, for $(a,b) \in T\times T^{\op}$
	\[
	\((a,b)^{\flip}\)^{\flip}=(a,b)~.
	\] 
	Let $(a,b),(c,d) \in T\times T^{\op}$. Then,
	\begin{align*}
		((a,b)(c,d))^{\flip}&=(ac,db)^{\flip}\\
		&=(db,ac)\\
		&=(d,c)(b,a)\\
		&=(c,d)^{flip}(a,b)^{\flip}~.
	\end{align*}                 
	Hence $\flip$ is an involution on $T \times T^{\op}$.
\end{proof}

\begin{proposition}
\label{prop:sem-to-invsem}
If a language is recognised by a semigroup, then it is recognised by a $\star$-semigroup as well.
\end{proposition}

\begin{proof}
	Let $\scrA = (A, \da)$ be an involutory alphabet. Assume that $L$ is recognised by $S$ using the morphism $h\colon A^{+} \rightarrow S$. By Lemma \ref{lemma:h-star}, $h^{\da}\colon A^{+} \rightarrow S^{\op}$ is a morphism. By Lemma \ref{lemma:s-s-op}, $(S \times S^{\op}, \flip)$ is an involution semigroup. Then, the product map $g\colon w \mapsto (h(w), h^{\da}(w))$, for $w \in A^{+}$, is a morphism from $A^{+}$ to $S\times S^{op}$. 
	Since 
	\[
	g(w^{\da})=(h(w^{\da}),h^{\da}(w^{\da}))=(h^{\da}(w),h(w))=(g(w))^{\flip}~,
	\]
	$g$ is a morphism of involution semigroups from $\scrA^{+}$ to $S\times S^\op$. Finally, the proof is completed by observing that $L$ is recognised by the involution semigroup $S\times S^{op}$ with the accepting set $\{ (h(w),h^{\da}(w)) \mid w \in L\}$.
\end{proof}

Therefore, languages recognised by involution semigroups are precisely the class of recognisable languages.

 Let $\scrS = (S, \star)$  be an involution semigroup. 
If $S' \subseteq S$, then $S'^{\star}$ denotes the set $\{ s^{\star} \mid s\in S'\}$.
A subset $S'$ of $S$ forms a \emph{sub-$\star$-semigroup} of $\scrS$ if 
$S'$ is closed under the semigroup operation and the involution, i.e., if $S'^{2} \subseteq S'$ and $S'^{\star} = S'$. If $h\colon\scrS \rightarrow \scrT$ is an involution semigroup morphism, then the image of $h$ is a sub-$\star$-semigroup of $\scrT$.

An element $s \in S$ is \emph{hermitian} if it is  its own involution, i.e., ~if $s^{\star}=s$. A $\star$-semigroup is hermitian-generated if it is generated by its hermitian elements using the semigroup operation and the involution. 
For convenience, we shorten hermitian-generated involution semigroup to \emph{hermitian semigroup}. For every hermitian alphabet $\scrA$, the involution semigroup $\scrA^{+}$ is hermitian. Images of hermitian semigroups under involution semigroup morphisms are hermitian as well.

\begin{example}
In the involution semigroup $\scrT = (T, \star)$ in Example \ref{example:inv}, the set of hermitian elements is $\{ab,ba\}$. Since they don't generate all the elements of the semigroup, $\scrT$ is not a hermitian semigroup.  But $(\{ab, ba\}, \star)$ is a sub-$\star$-semigroup of $\scrT$ that is hermitian.
\end{example}

For recognising languages that are subsets of $A^{*}$, the notion of recognisability by $\star$-semigroups is insufficient; one needs recognisability by $\star$-monoids. The definitions are straightforward adaptations of the ones presented above.


\section{Involutory Semidirect Products and The Characterisation of FO}
\label{Sec:4} 

In this section, we first define the involutory semidirect product of  two involution semigroups. Then, a particular case, called locally hermitian semidirect product, is defined. It is then shown that $\lrtt$ languages are precisely the ones recognised by locally hermitian products of aperiodic commutative $\star$-semigroups and locally trivial $\star$-semigroups.

\subsection{Bilateral Semidirect Products of Involution Semigroups}

Let $S$ and $T$ be finite semigroups. We denote the semigroup operation of $S$ additively (by $+$) and of $T$ multiplicatively.
However, we don't assume that $S$ or $T$ is commutative.
 A \emph{left action  $l$ of $T$ on $S$} is a map $l\colon T \times S \rightarrow S$ satisfying the following conditions. Denote the image of $(t, s)$ under $l$ as $ts$,  for $t\in T, s\in S$. Then, for all $s \in S$ and $t,t' \in T$,
\begin{align}
t(s+s') &= ts + ts'~, \text{ and},\\
(tt')s &= t (t's)~.
\end{align}

\emph{Right actions} of $T$ on $S$ are defined analogously; a \emph{right action  $r$ of $T$ on $S$} is a map $r\colon S \times T \rightarrow S$ such that for all $s \in S$ and $t,t' \in T$,
\begin{align}
(s+s')t &= st + s't~, \text{ and},\\
s(tt') &= (st)t'~,
\end{align}
where $st$ denotes the image of $(s,t)$ under the map $r$.

The actions $l$ and $r$ are {\em compatible} if for all $t, t' \in T$, and $s \in S$, it is the case that 
\begin{align}\label{eqn:compatibility}
(ts)t' = t(st')~. 
\end{align}

A pair of compatible actions $(l,r)$ is called a \emph{bilateral action} of $T$ on $S$.
Given a bilateral action $(l,r)$ of $T$ on $S$, the \emph{bilateral semidirect product} (also called the {\em two-sided semidirect product}) $S\sdp T$  is the semigroup with the 
elements $\{ (s, t) \mid s\in S, t \in T\}$ and with the operation 
\begin{align}\label{eqn:sdp}
(s_1 ,t_1)(s_2,t_2)=(s_1t_2+ t_1 s_2, t_1 t_2)~.
\end{align}

Let $\scrS = (S, \star)$ and $\scrT =(T, \diamond)$ be involution semigroups. Let $(l,r)$ be a bilateral action of $T$ on $S$. Then, the pair $(l,r)$ is an \emph{involutory action} of $\scrT$ on $\scrS$ if for all $t,t' \in T, s\in S$, 
\begin{equation}\label{eqn:inv-compatible-eqn}
 (st)^\star = t^{\diamond}s^\star~.  \\
\end{equation} 
 
\begin{lemma}
\label{lemma:inv-sem-verify}
Assume $\scrT = (T, \diamond)$ has an involutory action on $\scrS = (S, \star)$. Then, for all $t,t' \in T, s\in S$, 
\[
(t_{1}st_{2})^\star  = t_{2}^{\diamond}s^\star t_{1}^{\diamond}~.
\]
\end{lemma}

\begin{proof}
	By taking $\star$ on both sides in (\ref{eqn:inv-compatible-eqn}), we get
	\begin{equation}\label{eqn:inv-compatible-eqn1}
		st = (t^{\diamond}s^\star)^{\star}~.  \\
	\end{equation} 
	Exchanging $t$ and $t^{\diamond}$, and $s$ and $s^{\star}$ in (\ref{eqn:inv-compatible-eqn1}), we obtain
	\begin{equation}\label{eqn:inv-compatible-eqn2}
		s^{\star}t^{\diamond} = (ts)^{\star}~.  \\
	\end{equation} 
	Finally,
	\begin{align*}
		(t_{1}st_{2})^\star  &= ((t_{1}s)t_{2})^\star && \text{by (\ref{eqn:compatibility})} \\
		&= t_{2}^{\diamond} (t_{1}s)^{\star} && \text{by (\ref{eqn:inv-compatible-eqn})}\\
		&= t_{2}^{\diamond} s^{\star} t_{1}^{\diamond}~. && \text{by (\ref{eqn:inv-compatible-eqn2})}
	\end{align*}
\end{proof}

\begin{definition}
Assume  $\scrT=(T,\diamond)$ has an involutory action on $\scrS=(S, \star)$. Then, their \emph{involutory semidirect product} $\scrS \sdp \scrT$ is the semigroup $S \sdp T$ with the involution 
\begin{equation}\label{eqn:inv-sdp}
\dag\colon (s,t) \mapsto (s^\star, t^\diamond)~.
\end{equation}
\end{definition}

\begin{lemma}
\label{lemma:sdp-inv-semigroup}
 $\scrS \sdp \scrT$ is an involution semigroup.
\end{lemma}

\begin{proof}
	It suffices to show that $\dag$ is an involution.
	By definition, 
	\[
	{(s,t)^\dag}^\dag = ({s^{\star}}^{\star},{t^{\diamond}}^{\diamond})=(s,t)~.
	\]
	Next, for all $s_1, s_2 \in S, t_1,t_2\in T$,
	\begin{align*}
		\((s_1, t_1)(s_2,t_2)\)^\dag &= (s_1t_2+ t_1 s_2, t_1 t_2)^\dag && \text{by (\ref{eqn:sdp})}\\
		&= \((s_1t_2+ t_1 s_2)^\star, (t_1 t_2)^\diamond\)  && \text{by (\ref{eqn:inv-sdp})}\\
		&= \((t_1 s_2)^\star + (s_1t_2)^\star, t_2^\diamond t_1^\diamond\)\\
		&= \(s_2^\star t_1^\diamond + t_2^\diamond s_1^\star, t_2^\diamond t_1^\diamond\)\\
		&= \(s_2^\star, t_2^\diamond\)\(s_1^\star, t_1^\diamond\)\\
		&= \(s_2, t_2\)^\dag \(s_1, t_1\)^\dag~.
	\end{align*}
	Hence $\dag$ is an involution.
\end{proof}

Next, we introduce a particular case of involutory semidirect product. 

An element $e$ in a semigroup is an \emph{idempotent} if $e\cdot e = e$. In the case of involution semigroups, if $e$ is an idempotent, then involution of $e$ is also an idempotent.
\begin{definition}
Let $\scrS = (S, \star)$ and $\scrT =(T, \diamond)$ be involution semigroups. Let $(l,r)$ be an involutory action of $\scrT$ on $\scrS$. The action is \emph{locally hermitian} if for each idempotent $e \in T$, and each element $s\in S$, 
\begin{align} \label{eqn:local}
ese^\diamond &= e s^\star e^\diamond~.
\end{align}
\end{definition}

In other words, in the case of a locally hermitian action, the elements of the form $ese^{\diamond}$ are hermitian
(observe that, by Lemma \ref{lemma:inv-sem-verify}, $\(ese^{\diamond}\)^{\star} = 
{e^{\diamond}}^{\diamond}s^{\star}e^{\diamond}=es^{\star}e^{\diamond}$).
 
Let $\scrS$ and $\scrT$ be involution semigroups. If the action of $\scrT$ is locally hermitian action on $\scrS$, then the semidirect product $\scrS \sdp \scrT$ is called \emph{locally hermitian}.

\subsection{Characterisation of FO with Neighbour}

A semigroup $S$ is \emph{aperiodic} if there is an $n\in \bbN$ such that $a^{n+1}=a^{n}$, for all elements $a$ in the semigroup. A semigroup is \emph{commutative} if the semigroup operation is commutative, i.e., $x\cdot y = y \cdot x$, for all elements $x$ and $y$. A semigroup $S$ is locally trivial if $ese=e$, for each element $s$ and each idempotent $e$ in $S$. The class of languages $\ltt$ has several characterisations in terms of various products of semigroups (see \cite{BeauquierPin} for a detailed discussion); the one relevant to us is,

\begin{proposition}
\label{prop:ltt-char}
A language is in $\ltt$ if and only if it is recognised by a bilateral semidirect product of an aperiodic commutative semigroup and a locally trivial semigroup.
\end{proposition}
 Using the construction outlined in the proof of Proposition \ref{prop:sem-to-invsem}, it can be shown that a language is in $\ltt$ if and only if it is recognised by an involutory semidirect product of an aperiodic commutative involution semigroup and a locally trivial involution semigroup. 
 
Next, we characterise the class $\lrtt$. First, we give an example.

\begin{example}
Consider the language $L = (abc)^{+}$ over the alphabet $A=\{a,b,c\}$. We show that $L$ is defined by a locally hermitian product of two $\star$-semigroups $\scrS$ and $\scrT$. 

Let $T$ be the semigroup with the set $A^{\leq 2} \setminus \{ \epsilon \}$ and the operation $\odot$ given by, for $x,y \in T$
\[
x \odot y = \begin{cases}
x\cdot y & \text{ if $|xy| \leq 2$}\\
\mathit{first}(xy)\cdot \mathit{last}(xy) & \text{ otherwise},
\end{cases}
\]
where $\mathit{first}(w)$ and $\mathit{last}(w)$ denote the first and last letter of a word $w$. Reverse operation is an involution on the semigroup $T$. Let $\scrT = (T, r)$. Idempotents of $T$ are precisely those words of length $2$.

Let $\hat{A} = \{ \hat{a}, \hat{b}, \hat{c}\}$. Let $U$ be the set $A^{\leq 1}\cdot \hat{A}\cdot  A^{\leq 1}$. 
Elements of $T$ act on words in $U$ in the following way: for $w\in T$ and $x\hat{y}z \in U$, where $x,z \in A^{\leq 1}$ and $\hat{y}\in \hat{A}$,  we define the right action  $\otimes$ as
\[
w \otimes  x\hat{y}z= \begin{cases}
\mathit{last}(w)\cdot \hat{y} z & \text{ if $x$ is $\epsilon$}\\
 x\hat{y}z & \text{ otherwise},
\end{cases}
\]
The right action is analogous.
\[
x\hat{y}z \otimes w = \begin{cases}
x  \hat{y} \cdot \mathit{first}(w) & \text{ if $z$ is $\epsilon$}\\
 x\hat{y}z & \text{ otherwise},
\end{cases}
\]

Let $S'$ be the semigroup formed by the powerset of $U$ with set union as the semigroup operation. For each element $I\subseteq U$ of $S'$,  let $I^{r} = \{ w^{r} \mid w \in I\}$. The operation $I \mapsto I^{r}$, for $I\in S'$, is an involution on the semigroup $S'$. 
Let $\scrS'=(S', r)$. The action of $T$ on $U$, extends to the elements of $\scrS'$ pointwise. The elements $A\cdot \hat{A}\cdot  A$ are called the \emph{zeros} for the action. 
Let $\sim$ be a congruence on $\scrS'$, given by, for $I,J \in S'$,
\begin{itemize}
\item $x \in I$ if and only if $x\in J$, for every nonzero $x \in U$, and
\item $x$ or $x^{r}$ is in $I$ if and only if $x$ or $x^{r}$ is in $J$,  for each zero $x \in U$.
\end{itemize}
Let $\scrS$ be the quotient of $\scrS'$ with respect to $\sim$. 
It is easy to see that for each idempotent $e \in T$, and each $I \in \scrS$, $e\otimes I \otimes e^{r} = e\otimes I^{r} \otimes e^{r}$. Hence the actions are locally hermitian. Finally $L$ is accepted by $\scrS \sdp \scrT$ with the morphism $h\colon a \mapsto \(\{\hat{a}\}, a\)$.

\end{example}
 \begin{theorem}
\label{lemma:lrtt-acomli}
Each $\lrtt$ language is recognised by a locally hermitian semidirect product of an aperiodic commutative $\star$-semigroup and a locally trivial $\star$-semigroup.
Conversely, if $L$ is recognised by a locally hermitian semidirect product  $\scrS \sdp \scrT$,  where $\scrS$ is aperiodic and commutative and  $\scrT$ is locally trivial, then $L$ is $\lrtt$.
\end{theorem}

\begin{proof}

	Assume $L \subseteq A^{+}$, for an involutory alphabet $\scrA = (A, \id)$, is a $\lrtt$ language. To prove the first claim, we define $\star$-semigroups $\scrS$ and $\scrT$ such that $\scrS$ is aperiodic  and commutative and $\scrT$ is locally trivial, and show that their locally hermitian semidirect product recognises $L$. 
	We let $k, m>0$ be such that $L$ is a union of $\approxr{m}{2k+1}$-classes. 
	Since the equivalence $\approxr{t}{d+1}$ refines the equivalence  $\approxr{t}{d}$, for $t,d>0$, there is no loss of generality.

	Let $\odot$ be the operation on the set $A^{\leq 2k} \setminus \{\epsilon\}$ defined as

	$$t\odot t' = \begin{cases} t\cdot t' & \text{if $|tt'| < 2k$,}\\
		\pref_k(t\cdot t')\cdot \suff_k(t\cdot t')  & \text{otherwise,}
	\end{cases}$$
	where $\pref_k()$ and $\suff_k()$ denote the $k$-length prefix and suffix respectively. 
	It is not difficult to see that the operation $\odot$ is associative. Hence, $T=(A^{\leq 2k}\setminus \{\epsilon\}, \odot)$ is a semigroup. 
	It is straight-forward to check that $(u \odot v)^{r}=v^{r}\odot u^{r}$ and $(u^{r})^{r}=u$, for the reverse operation $r$. 
	Let $\scrT=(T, r)$ be the involution semigroup with the reverse operation as the involution. The idempotents of $T$ are precisely the words of length $2k$. Also, $e\odot t \odot e=e$, for words $t \in T$ and $e\in A^{2k}$. Hence, $\scrT$ is locally trivial. 
	
	Next, we set up some notation in preparation for defining $\scrS$. Let $\hat{A}$ denote the set $\{\hat{a} \mid a \in A\}$ of \emph{anchored letters}. A $k$-{\em anchored word} $\hat{w}$ is a word from the set $A^{\leq k}  \cdot \hat{A} \cdot A^{\leq k}$. Let $W_{k}$ denote the set of all $k$-anchored words. We denote anchored words by $\hat{w}, \hat{u}$ etc.
	
	Let $\hat{w} = u\hat{a}v$, where $u,v \in A^{*}$ and $\hat{a} \in \hat{A}$.
	We write $l(\hat{w})$ and $r(\hat{w})$ to denote how far is the anchored letter from the left-end and right-end respectively; Hence, $l(\hat{w}) = |u|$ and $r(\hat{w}) = |v|$. Clearly, $0 \leq l(\hat{w}), r(\hat{w}) \leq k$, for all $\hat{w} \in W_{k}$.
	
	The semigroup $T$ acts on the set $W_{k}$, on the left as well as on the right, in the following manner. For each letter $a\in A$ and $\hat{w} \in W_{k}$, the left action $a\otimes \hat{w} \in W_{k}$ of $a$ on  $\hat{w}$ is given by,
	\begin{equation}
		a\otimes \hat{w} = \begin{cases}
			\hat{w} & \text{if $l(\hat{w})=k$,}\\
			a\cdot \hat{w} & \text{otherwise.}
		\end{cases}
	\end{equation}
	
	Similarly the right action $\hat{w}\otimes a $ appends $a$ to the right if there are less than $k$ positions to the right of the anchored position, otherwise it leaves $\hat{w}$ intact. We extend these actions to all the elements of $T$ by letting, for $u \in T$, $a \in A$ and $\hat{w} \in W_{k}$,
	\begin{equation}
		\label{eq:s1}
		(u\odot a) \otimes \hat{w} = u \otimes (a \otimes \hat{w})~, \text{ and } \hat{w} \otimes (a\odot u)  = (\hat{w} \otimes a) \otimes u~.
	\end{equation}
	
	Clearly the left and right actions are compatible with each other, i.e., 
	\begin{equation}
		\label{eq:s2}
		(u \otimes \hat{w}) \otimes v = u \otimes (\hat{w} \otimes v)
	\end{equation}
	for all $\hat{w}\in W_{k}$ and $u,v \in T$. 
	Also the actions are involutory, with respect to the involution on $T$, since 
	\begin{equation}
		\label{eq:s3}
		(u \otimes \hat{w})^r  = \hat{w}^r \otimes u^r~.
	\end{equation}

	An anchored word is a {\em zero} for the action $\otimes$ if for all $a \in A$, $a\otimes \hat{w} = \hat{w} \otimes a = \hat{w}$. The zero elements of $W_{k}$ are precisely the anchored words whose anchored letter is at distance $k$ from each end.
	
	Next, we define the semigroup $S$. A \emph{multiset} in the universe $U$, where $U$ is a set, is a collection of elements from $U$ that allows for multiple instances. For instance, the multiset $\{3, 3, 4\}$, in the universe $\bbN$, contains two occurrences of $3$ and one occurrence of $4$. Given a multiset $I$ and an element $x\in I$, the \emph{multiplicity} of $x$ in $I$, denoted as $m_{I}(x)$, is the number of times $x$ occurs in $I$. The \emph{support} of $I$, denoted as $\supp(I)$, is the set of all elements occurring in $I$.  Given two multisets $I$ and $J$, the \emph{union} of $I$ and $J$, denoted as $I\cup J$, is the multiset $K$ with the support $\supp(I) \cup \supp(J)$ and multiplicity $m_{K}(x) = m_{I}(x) + m_{J}(x)$, for all $x \in \supp(K)$. Let $M(U)$ denote the set of all \emph{non-empty} multisets in the universe $U$. The set $M(U)$ is a semigroup with $\cup$ as the semigroup operation.
	
	Consider the semigroup $M(W_{k})$. Clearly, $M(W_{k})$ is infinite. 
	For each multiset $I$ in $M(W_{k})$, we let $I^{r}$ denote the multiset with support $\{ x^{r} \mid x \in I \}$ and multiplicity $m_{I^{r}}(x^{r})=m_{I}(x)$, for each $x\in I$. The reverse operation is an involution on $M(W_{k})$.

	We define the following equivalence relation $\equiv^{m}$ on the multisets in $M(W_{k})$; for $I, J$ in $M(W_{k})$, $I \equiv^{m} J$ if
	\begin{enumerate}
		\item  $m_{I}(x) =^{m}  m_{J}(x)$, for each non-zero element $x$ of $W_{k}$, and,
		\item $m_{I}(x)+m_{I}(x^{r}) =^{m}  m_{J}(x) + m_{J}(x^{r})$, for each zero element $x$ of $W_{k}$.
	\end{enumerate}
	
	The equivalence relation $\equiv^{m}$ is a congruence relation on the involution semigroup $M(W_{k})$: If $I \equiv^{m} J$, then  $I^{r}  \equiv^{m} J^{r}$, and also, if $I \equiv^{m} J$ and $K \equiv^{m} L$, then $I \cup K \equiv^{m} J \cup L$, for all $I,J,K,L \in M(W_{k})$. 
	
	We let $S$ to be the quotient $M(W_{k}) \slash \equiv^{m}$. Clearly, $S$ is a finite semigroup.
	We denote the semigroup operation on $S$ additively. 
	The $\star$-semigroup $\scrS = (S, r)$ is aperiodic and commutative.
	
	We can extend the involutory action $\otimes$ of $\scrT$ on $W_{k}$, to the $\star$-semigroup $\scrS$ pointwise, i.e., we let 
	the left action $u \otimes I$, for $u \in T$ and $I \in S$, to be the {multiset} of anchored words with the support $\{ u\otimes \hat{w} \mid \hat{w} \in I\}$ and multiplicity 
	\begin{equation}
		m_{u \otimes I}(\hat{v}) = \sum_{\hat{w} \in I, u \otimes \hat{w} = \hat{v}} m_{I}(\hat{w})~,
	\end{equation}
	for $\hat{v} \in W_{k}$.
	Right action $I\otimes u$ is defined dually. By Equations \ref{eq:s1}, \ref{eq:s2}, \ref{eq:s3}, the actions of $T$ on $S$ are involutory. 
	It is easily observed that, by definition of the action $\otimes$ and the congruence $\equiv^{m}$, for each idempotent $e \in T$ and each multiset $I \in S$,
	\begin{equation}
		e \otimes I \otimes e^{r}  = e \otimes I^{r} \otimes e^{r}~.
	\end{equation}
	Hence the action of $\scrT$ on $\scrS$ is locally hermitian. For convenience, below we identify the anchored word $\hat{w}$ and the singleton multiset $\{ \hat{w}\}$. Also, we will omit writing the action $\otimes$ explicitly.

	Let $\scrS \sdp \scrT$ be the locally hermitian semidirect product of $\scrS$ and $\scrT$ with the action $\otimes$. We define the $\star$-morphism $h$ from $\scrA^+$ to $\scrS \sdp \scrT$ as 
	\[
	h\colon a \mapsto (\hat{a}, a)\in S\times T~.
	\]
	It is easy to verify that $h$ is a morphism of involution semigroups. 
	
	Next, we show that the morphism $h$ recognises the language $L$; it suffices to show that if $w \approxr{m}{2k+1} w'$, then $h(w) = h(w')$.

	Assume $w=a_1\cdots a_n \in A^{+}$ and $w'=a'_1\cdots a'_\ell$ are two words such that
	$w \approxr{m}{2k+1} w'$. If $|w|<2k+1$ or $|w'|<2k+1$, then clearly $w=w'$ and $h(w)=h(w')$. Therefore, assume that both words have length at least $2k+1$.
	The image of $w$ under $h$ is of the form
	\begin{equation}
		\begin{split}
			h(w) &= (\hat{a_1}, a_1) \cdots (\hat{a_n}, a_n)\\
			&= (\hat{a_1} a_2 \cdots  a_n +  a_1 \hat{a_2}\cdots a_{n}  +\cdots+ a_1 \cdots a_{n-1} \hat{a_n},\\
			&\quad~~ a_1 \cdots a_n)~.
		\end{split}
	\end{equation}
	Observe that if $u,v \in T$ are such that $|u|, |v| \geq k$, then $u \otimes \hat{w} \otimes v = \suff_{k}(u) \otimes \hat{w} \otimes \pref_{k}(v)$, and $u \odot v =  \pref_{k}(u) \cdot \suff_{k}(v)$, for each $\hat{w} \in S$. Since $n\ge2k+1$, we derive, 
	\begin{align*}
		h(w) = (&\hat{a_1} a_2 \cdots  a_{k+1} +  a_1 \hat{a_2}\cdots a_{k+2}  +\cdots\\
		&+ a_{n-k} \cdots a_{n-1} \hat{a_n},~~ a_1 \cdots a_k a_{n-k+1} \cdots a_n)~.
	\end{align*}
	
	Similarly for $w'=a'_1\cdots a'_\ell$ the image  is of the form
	\begin{align*}
		h(w')  = (&\hat{a_1'} a_2' \cdots  a_{k+1}' +  a'_1 \hat{a'_2}\cdots a'_{k+2}  +\cdots\\
		&+ a'_{\ell-k} \cdots a'_{\ell-1} \hat{a'_\ell},~~ a'_1 \cdots a'_k a'_{\ell-k+1} \cdots a'_\ell)
	\end{align*} 
	
	Since $w$ and $w'$ have a common prefix and suffix of length $2k$, their images in $T$ are identical. Hence, all it remains to conclude the claim is to show that the summations 
	\begin{align*} 
		L &= \hat{a_1} a_2 \cdots  a_{k+1} +  a_1 \hat{a_2}\cdots a_{k+2}  +\cdots+ a_{n-k} \cdots a_{n-1} \hat{a_n}\\
		R &= \hat{a_1'} a_2' \cdots  a_{k+1}' +  a'_1 \hat{a'_2}\cdots a'_{k+2}  +\cdots+ a'_{\ell-k} \cdots a'_{\ell-1} \hat{a'_\ell}
	\end{align*}
	have the same value in $S$. Observe that, since $S$ is aperiodic and commutative, we can freely commute the terms in $L$ and $R$.
	
	The zero elements in the summation $L$ are precisely
	the terms of the form 
	\[a_{i-k}\cdots a_{i-1}\hat{a_i} a_{i+1}\cdots a_{i+k}~,
	\]
	
	where $k<i\leq n-k$.
	We observe that there is a one-to-one correspondence between the factors of $w$ of length $2k+1$ and zero elements in the summation $L$; a factor of $w$ of length $2k+1$ centred at a position $k<i\leq n-k$, corresponds to the term $a_{i-k}\cdots a_{i-1}\hat{a_i} a_{i+1}\cdots a_{i+k}$. Similarly, there is a one-to-one correspondence between the factors of $w'$ of length $2k+1$ and zero elements in the summation $R$. The non-zero elements in $L$ and $R$ are identical, owing to the fact that $w$ and $w'$ have the same  prefix and suffix of length $2k$. Finally, since sets of $2k+1$-factors of $w$ and $w'$ are equal up to reverse and up to threshold $m$, it mean that sets of zeros of $L$ and $R$ are also identical up to reverse and threshold $m$. Since in the semigroup $\scrS$, by definition of $\equiv^{m}$, zero elements are identified with their reverse and they are counted only up to threshold $m$, we deduce that $L=R$. Hence the claim is proved.

	Next we prove the converse direction.

	Let $\scrA=(A, \id)$ be an involutory alphabet.
	Assume $L \subseteq A^{+}$ is recognised by the locally hermitian product $\scrS \sdp \scrT$, for some aperiodic commutative $\star$-semigroup $\scrS = (S, \star) $ and locally trivial $\star$-semigroup $\scrT = (T, \diamond)$, with the morphism $h\colon\scrA^{+}\rightarrow \scrS \sdp \scrT$ and the accepting set $P$. 
	
		Note that since $\scrA$ is hermitian, for each $(s, t) \in S \times T$ such that $(s,t)=h(a)$ for some $a\in A$,
	\begin{equation}
	\label{Eq:himage}
	(s,t)= (s^{\star}, t^{\diamond})~.
	\end{equation}

	We use the following property of locally trivial semigroups: If $T$ is locally trivial, then there is a $k\geq 1$ such that
	\begin{equation} 
		\label{eq:locally-trivial}
		x_1\cdots x_k\cdot z \cdot  y_1\cdots y_k = x_1 \cdots x_k \cdot y_1\cdots y_k \text{ for all $x_i, z, y_i \in T$.}
	\end{equation}
	Let $k$ be such a number for the locally trivial semigroup $T$.
	
	Since the semigroup $S$ is aperiodic there exists a natural number $t>0$, such that $s^t=s^{t+1}$ for all elements $s\in S$.

	Let $w, w' \in \scrA^+$ be two words such that $w \approxr{t}{4k+1} w'$. We claim that $w \in L$ if and only if $w' \in L$; proving the claim implies the lemma.
	We assume that $w$ and $w'$ are of length at least $4k+1$, otherwise $w=w'$, and the claim is immediate.
	
	Let $w = a_1 \cdots a_n \in \scrA^+$ be a word. Let $h(a_i) = (s_i, t_i)$. 
	The image of $w$ under $h$ is
	\begin{equation}
		\begin{split}
			h(w) &= (s_1, t_1) \cdots (s_n, t_n)\\
			&= (s_1 t_2 \cdots  t_n  + t_1 s_2 t_{3} \cdots  t_n +\cdots
			+ t_1 \cdots t_{n-1} s_n,\\
			&\quad~~ t_1 \cdots t_n)~.
		\end{split}
	\end{equation}

	Similarly, let $w' = a_1' \cdots a_m' \in \scrA^+$. Assuming that $h(a_{i}')=(s_{i}',t_{i}')$, the image of $w'$ under $h$	 is
	\begin{equation}
		\begin{split}
			h(w') &=  (s_1' t_2' \cdots  t_m'  + t_1' s_2' t_{3}' \cdots  t_m' +\cdots+ t_1' \cdots t_{m-1}' s_m',\\
			&\quad~~ t_1' \cdots t_m')~.
		\end{split}
	\end{equation}
	
	Since $m,n \ge 4k+1$ and $w \approxr{t}{4k+1} w'$, by definition, the words $w$ and $w'$ have the same prefix and suffix of length $4k$. Using Equation \ref{eq:locally-trivial}, we conclude that $t_1 \cdots t_n = t_1' \cdots t_m'$. Therefore, it only remains to show that the summations 
	\begin{align*} 
		L &= s_1 t_2 \cdots  t_n  + t_1 s_2 t_{3} \cdots  t_n +\cdots+ t_1 \cdots t_{n-1} s_n\\
		R &= s_1' t_2' \cdots  t_m'  + t_1' s_2' t_{3}' \cdots  t_m' +\cdots+ t_1' \cdots t_{m-1}' s_m'
	\end{align*}
	have the same value in $S$. Observe that, since $S$ is aperiodic and commutative, we can freely commute the terms in $L$ and $R$.
	Let $p_{i}=t_{1}\cdots t_{i-1}$ and $q_{i}=t_{i+1}\cdots t_{n}$, for $1\leq i \leq n$. Then, 
	\[
	L = \sum_{i=1}^{n} p_{i}s_{i}q_{i} 
	\]
	Similarly, letting $p'_{i}=t'_{1}\cdots t'_{i-1}$ and $q'_{i}=t'_{i+1}\cdots t'_{m}$, for $1\leq i \leq m$. Then, 
	\[
	R = \sum_{i=1}^{m} p'_{i}s'_{i}q'_{i} 
	\]
	We are going to consider $p_{i},q_{i},p_{i}',q_{i}'$ etc.~as words over the alphabet $T$ and manipulate them. Let $f\colon T^{+} \rightarrow T^{+}$ be the function 
	defined as, for each $w \in T^{+}$,
	\[
	f(w) = \begin{cases}
		w & \text{if $|w|<2k$},\\
		\pref_{k}(w)\cdot \suff_{k}(w) & \text{otherwise}.
	\end{cases}
	\]
	
	Using Equation \ref{eq:locally-trivial}, the expressions $L$ and $R$ can be rewritten as
	\begin{align}
		\label{eqn:auxnormalform}
		L &= \sum_{i=1}^{n} f(p_{i})\,s_{i}\,f(q_{i})~,\\
		R &= \sum_{i=1}^{m} f(p'_{i})\,s'_{i}\,f(q'_{i})~.
	\end{align}
	
	Let $L_{i}$ denote the term $f(p_{i})\,s_{i}\,f(q_{i})$. Similarly, let 
	$R_{i}$ denote the term $f(p'_{i})\,s'_{i}\,f(q'_{i})$.
	
	Since the words $w$ and $w'$ have the same prefix of length $4k$, the partial sums
	\[
	\sum_{i=1}^{2k-1} L_{i} \text{ and } \sum_{i=1}^{2k-1} R_{i}
	\]
	are identical. Similarly, since the words have the same suffix of length $4k$, the partial sums
	\[
	\sum_{i=n-2k+1}^{n}L_{i} ~\text{ and }~ \sum_{i=m-2k+1}^{m}R_{i}
	\]
	are also identical. Thus, all it remains is to show is the equality of the following partial sums:
	\[
	L_{f}=\sum_{i=2k}^{n-2k} L_{i} \text{ and }R_{f}=\sum_{i=2k}^{m-2k} R_{i}~.
	\]
	The terms occurring in $L_{f}$ and $R_{f}$ are called \emph{special}.
	
	Let $x=t_{1}\cdots t_{k} = t'_{1}\cdots t'_{k}$ and $y=t_{n-k+1}\cdots t_{n}= t'_{m-k+1}\cdots t'_{m}$. Observe that, for each special term $L_{i}$, 
	\[
	L_{i}= x \cdot (t_{i-k}\cdots t_{i-1}) \, s_{i} (t_{i+1}\cdots t_{i+1+k}) \cdot y~.
	\]
	
	Similarly, for each special term $R_{i}$,
	\[
	R_{i}= x\cdot (t'_{i-k}\cdots t'_{i-1}) \, s'_{i} (t'_{i+1}\cdots t'_{i+1+k}) \cdot y~.\]

	Before proceeding, we make an observation about special terms.
	Let $e$ and $e^{\diamond}$ be two idempotents in $T$. Let $p,q,u,v \in T^{+}$ be such that $|p|,|q|,|u|,|v| \geq k$, and let $s \in S$, then,
	\begin{align}
		pq s uv &= p e (q s u) e^\diamond v  && \text{(using  Equation \ref{eq:locally-trivial})}\\
		&= p e (q s u)^\star e^\diamond v  && \text{($\because$ the actions are locally hermitian)}\\
		&= p e u^\diamond s^\star q^\diamond e^\diamond v  && \text{($\because$ the actions are involutory)}\\
		&= p u^\diamond s^\star q^\diamond v~.  && \text{(using  Equation \ref{eq:locally-trivial}) \label{eq:hermitian-rotate}}
	\end{align}
	
	Applying the above observation and Equation \ref{Eq:himage} to the term $L_{i}$ implies that 
	\begin{equation}
		\label{eq:herm-rotate1}
		\begin{split}
			x \cdot &(t_{i-k}\cdots t_{i-1}) \, s_{i} (t_{i+1}\cdots t_{i+1+k}) \cdot y \\
			&= x \cdot (t^{\diamond}_{i+1+k}\cdots t^{\diamond}_{i+1}) \, s_{i}^{\star} (t^{\diamond}_{i-1}\cdots t^{\diamond}_{i-k}) \cdot y\\
			&= x \cdot (t_{i+1+k}\cdots t_{i+1}) \, s_{i} (t_{i-1}\cdots t_{i-k}) \cdot y~.
		\end{split}
	\end{equation}
	Similarly for $R_{i}$,
	\begin{equation}
		\label{eq:herm-rotate2}
		\begin{split}
			x\cdot &(t'_{i-k}\cdots t'_{i-1}) \, s'_{i} (t'_{i+1}\cdots t'_{i+1+k}) \cdot y \\
			&= x \cdot (t'^{\diamond}_{i+1+k}\cdots t'^{\diamond}_{i+1}) \, {s'_{i}}^{\star} (t'^{\diamond}_{i-1}\cdots t'^{\diamond}_{i-k}) \cdot y\\
			&= x \cdot (t'_{i+1+k}\cdots t'_{i+1}) \, {s'_{i}} (t'_{i-1}\cdots t'_{i-k}) \cdot y\\
		\end{split}
	\end{equation}
	
	We are ready to conclude the claim $L_{f} = R_{f}$.  Factors of $w$, of length $2k+1$, and centred at positions $i$, for $2k \leq i <n-2k$, are called \emph{special}. Similarly, factors of $w'$, of length $2k+1$, and centred at  positions $i$, for $2k \leq i <m-2k$, are also called \emph{special}. 
	We observe that the value of each special term $L_{i}$ is completely determined by the special factor at $i$. A similar claim holds for the special term $R_{i}$ and the special factor of $w'$ at $i$. Since the semigroup $S$ is aperiodic, it is enough to compare the special factors in each word up to the threshold $t$. Moreover, by Equations \ref{eq:herm-rotate1} and \ref{eq:herm-rotate2}, it is enough to consider the special factors up to reverse. Finally, since $w \approxr{t}{4k+1} w'$, the claim follows.
\end{proof}

\section{Varieties of Involutory Languages}

\label{Sec:3} 

Next, we introduce the notion of varieties of involutory languages, and relate it to pseudovarieties of finite involution semigroups by proving an Eilenberg-type correspondence. For this purpose, we define the notion of a syntactic involution semigroup, and show that it possesses properties similar to that of syntactic semigroups. Finally, Theorem \ref{lemma:lrtt-acomli} is stated in terms of pseudovarieties of involution semigroups. 

It is possible to state the results in this section in a much more general setting using   categories, following \cite{mikolajmonad,mikolajmonadarxiv}, where 
the notions of recognisability, languages and language varieties etc.~are developed for 
monads and a variety theorem for them is proved. Variety theorems of many settings, for instance finite/infinite words, finite trees etc.~can be seen as instances of this unified result.
 In this framework, the involution operation is a monad that satisfies a certain distributive law. One advantage of this approach is that the variety theorem for involutary languages can be deduced from the general case. But it also requires the significantly more complex language of categories.  For the sake of simplicity we do the variety theorem in the usual way.

\subsection{Syntactic Involution Semigroups}

 The {\em syntactic congruence} of a language $L\subseteq A^+$, denoted as $\sim_L$, is the congruence relation on $A^+$, defined as 
\begin{equation}
x \sim_L y \quad := \quad  uxv \in L \text{ if and only if } uyv \in L, 
\end{equation} for all words $u,v \in A^*$.

The quotient semigroup $A^+\slash \sim_L$, denoted as $S(L)$, is called the {\em  syntactic semigroup} of $L$. 
Let $[u]_{{L}}$ denote the equivalence class of $u \in A^{+}$ under the congruence $\sim_{L}$. Then, the surjective morphism $\varphi_{L}\colon u \mapsto [u]_{{L}}$, for $u \in A^{+}$, from $A^+$ to $S(L)$, recognises $L$ with the accepting set $\varphi_{L}(L)$.  The morphism $\varphi_L$ is called the {\em syntactic morphism} of $L$.

It is possible to lift the notion of syntactic semigroup to the case of recognition by involution semigroups.

\begin{definition}
\label{defn:syn-star}
Let $\scrA=(A, \da)$ be an involutory alphabet. The {\em syntactic $\star$-congruence} of $L \subseteq A^{+}$ is the relation $\ssim_{L}$ on $A^+$, given by
\[
x \ssim_{L} y \quad := \quad x \sim_L y \text{ and } x \sim_{L^\da} y~. \\                                
\]
\end{definition}

Since $\ssim_{L}$ is an intersection of two congruence relations on $A^{+}$, it is a congruence on $A^{+}$. Hence, $A^+\slash \ssim_{L}$ is a semigroup. 

Let $[\![w]\!]_{{L}}$ denote the equivalence class of $w\in A^{+}$ under the congruence $\ssim_{L}$. Firstly, we observe

\begin{lemma}
\label{lemma:L-L-da}
Let $\scrA=(A, \da)$ be an involutory alphabet. Let $L \subseteq A^{+}$. 
Then, $x \sim_{L^{\da}} y$ if and only if $x^{\da} \sim_{L} y^{\da}$, for all $x,y \in A^{+}$.
\end{lemma}
\begin{proof}
	Assume $x \sim_{L^{\da}} y$, for $x,y \in A^{+}$. Let $u,v \in A^{*}$.
	Then, $uxv \in L^{\da}$ if and only if $uyv \in L^{\da}$, by definition of $\sim_{L^{\da}}$. It follows that 
	$v^{\da}x^{\da}u^{\da} \in L$ if and only if $v^{\da}y^{\da}u^{\da} \in L$. Since $u^{\da},v^{\da}$ are arbitrary, it follows that 
	$x^{\da} \sim_{L} y^{\da}$. The other directon is similar.
\end{proof}

\begin{lemma}
\label{lemma:syn-class-inv}
The map $[\![w]\!]_{{L}} \mapsto [\![w^{\da}]\!]_{{L}}$, for $w\in A^{+}$, is an involution on the equivalence classes of $\ssim_{L}$.
\end{lemma}
\begin{proof}
	It suffices to show that $x \ssim_{L} y$ if and only if $x^{\da} \ssim_{L} y^{\da}$, and 
	$(xy)^{\da} \ssim_{L} y^{\da}x^{\da}$, for all words $x,y \in A^{+}$.
	
	Firstly,
	\begin{align*}
		x \ssim_{L} y &\iff x \sim_L y \text{ and } x \sim_{L^\da} y  && \text{(by Definition \ref{defn:syn-star})} \\
		&\iff x^{\da} \sim_{L^{\da}} y^{\da} \text{ and } x^{\da} \sim_{L} y^{\da} && \text{(by Lemma \ref{lemma:L-L-da})}\\
		&\iff x^{\da} \ssim_{L} y^{\da}~. && \text{(by Definition \ref{defn:syn-star})}
	\end{align*}
	Lastly, since $(xy)^{\da} = y^{\da}x^{\da}$, $(xy)^{\da} \ssim_{L} y^{\da}x^{\da}$.
	
\end{proof}

\begin{definition}
The syntactic $\star$-semigroup of $L$ is the quotient semigroup $A^+\slash \ssim_{L}$ along with the involution $\da\colon [\![w]\!]_{{L}} \mapsto [\![w^{\da}]\!]_{{L}}$, for $w \in A^{+}$.
\end{definition}

 We can obtain the syntactic $\star$-semigroup of $L$ from the syntactic semigroup of $L$  as well. Let $\varphi_{L}^{\da}$ denote the morphism from $A^{+}$ to $S(L)^{\op}$ given by Lemma \ref{lemma:h-star}. Let $\chi_{L}$ be the product morphism from $A^{+}$ to $S(L) \times S(L)^{\op}$ given by, for $w \in A^{+}$,
\begin{equation}
\chi_{L}\colon w \mapsto \(\varphi_{L}(w), \varphi_{L}^{\da}(w)\)~.
\end{equation}
By definition of the morphism $\varphi_{L}^{\da}$,
\begin{equation}
\chi_{L}(w) = \([w]_{{L}}, [w^{\da}]_{{L}}\)~.
\end{equation}
It is easy to see that $\chi_{L}$ is a morphism from $\scrA^{+}$ to $\(S(L) \times S(L)^{\op}, \flip\)$.
Let $\S(L)$ denote the sub-$\star$-semigroup that is the image of the morphism $\chi_{L}$.

\begin{proposition}
\label{prop:syn-star-alt}
$\S(L)$ is isomorphic to the syntactic $\star$-semigroup of $L$.
\end{proposition}

\begin{proof}
	We claim that the map $f\colon [\![w]\!]_{{L}} \mapsto ([w]_{{L}}, [w^{\da}]_{{L}})$, for $w \in A^{+}$, is an isomorphism between $A^+\slash \ssim_{L}$ and $\S(L)$.
	If $w \ssim_{L} w'$, then $w \sim_{L} w'$ and $w \sim_{L^\da} w'$ by definition; by Lemma  \ref{lemma:L-L-da}, $w \sim_{L^\da} w'$ if and only if $w^{\da} \sim_{L} w'^\da$.
	Therefore, 
	the map  $f$ is well-defined. Clearly, it is surjective. By the definition of $\ssim_{L}$, it is injective as well. Since 
	\begin{align*}
		f([\![uv]\!]_{{L}}) &= ([uv]_{{L}},~ [(uv)^{\da}]_{{L}}) \\
		&= ([uv]_{{L}},~ [v^{\da}u^{\da}]_{{L}})\\
		&= ([u]_{{L}},~ [u^{\da}]_{{L}})([v]_{{L}}, [v^{\da}]_{{L}})\\
		&= f([\![u]\!]_{{L}})f([\![v]\!]_{{L}})~,
	\end{align*}
	$f$ is a morphism. Finally, since
	\begin{align*}
		f([\![w^{\da}]\!]_{{L}}) &= ([w^{\da}]_{{L}}, [(w^{\da})^{\da}]_{{L}})\\
		&= ([w^{\da}]_{{L}}, [w]_{{L}}) \\
		&= f\([\![w]\!]_{{L}}\)^{\flip}~,
	\end{align*}
	$f$ is a morphism of involution semigroups.
	Therefore $f$ is an isomorphism from $A^+\slash \ssim_{L}$ to $\scrS(L)$.
\end{proof}

As one would expect, the syntactic $\star$-semigroups $\S(L)$ and $\S(L^{\da})$ are closely related. They are anti-isomorphic to each other, i.e, $\S(L)$ is isomorphic to the opposite of $\S(L^{\da})$. For this, first we show
 
\begin{lemma}
\label{lemma:syn-L-L-da}
The map $\alpha\colon [w]_{{L}} \mapsto [w^{\da}]_{L^{\da}}$, for $w\in A^{+}$, is an isomorphism from $S(L)^{\op}$ to $S(L^{\da})$.
\end{lemma}

\begin{proof}
	By Lemma \ref{lemma:L-L-da}, the map is well-defined. Clearly, $\alpha$ is surjective. By the same lemma, it is injective.
	We verify that
	\begin{align*}
		\alpha([u]_{{L}}[v]_{{L}}) &= \alpha([vu]_{{L}}) && (\text{Product in $S(L)^{\op}$})\\
		&= [(vu)^{\da}]_{{L^{\da}}}\\
		&= [u^{\da}v^{\da}]_{{L^{\da}}}\\
		&= [u^{\da}]_{{L^{\da}}}[v^{\da}]_{{L^{\da}}}
	\end{align*}
	Hence, $\alpha$ is an isomorphism. 
\end{proof}

Therefore, the semigroup $S(L) \times S(L)^{\op}$ is isomorphic to $S(L) \times S(L^{\da})$. Likewise $S(L^{\da}) \times S(L^{\da})^{\op}$ is 
isomorphic to $S(L)^{\op} \times S(L)$. Finally, the operation $\flip$ is an isomorphism from $S(L) \times S(L)^{\op}$ to $S(L)^{\op} \times S(L)$. Therefore, $\S(L)$ and $\S(L^{\da})$ are anti-isomorphic to each other.

We look at some special cases. When $L^\da = L$, the semigroup reduct of $\S(L)$  is isomorphic to $S(L)$. Therefore, in the case of a reversible language, $\scrS(L)$ degenerates into the syntactic semigroup.

When $\scrA$ is hermitian, for every letter $a\in A$, the element $\chi_{L}(a) \in \S(L)$ is hermitian. Since such elements generate $\S(L)$, it is hermitian.

\begin{proposition}
\label{prop:syn-prop}
$\S(L)$ divides any $\star$-semigroup recognising $L$.
\end{proposition}
\begin{proof}
	\begin{figure}[h]
		\centering\SelectTips{cm}{}
		\leavevmode\xymatrix@C=60pt@R=30pt{
			\scrA^{+}  \ar[r]^{h}  \ar[dr]^{\chi_{L}} & \scrT \ar[d]^{(\mu, \mu^{\star})} \\
			& \S(L) \\
		}
		\caption{Universality of syntactic $\star$- semigroups}
		\label{fig:SynStar}
	\end{figure}
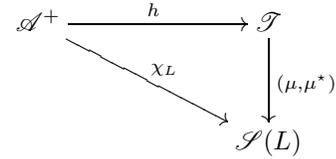
	
	The proof is analogous to the semigroup case \cite{pin1986varieties}, Proposition 2.7. Assume the language $L$ is recognised by a $\star$-semigroup $\scrT = {(T, \star)}$ with a $\star$-morphism $h \colon \scrA^{+}\rightarrow \scrT$.
	Without loss of generality we can assume that the map $h$ is surjective; otherwise we take $\scrT$ to be the image of $A^+$ under $h$, that is a sub-$\star$-semigroup of $\scrT$. 
	
	Since $T$ recognises the language $L$, there is a surjective morphism $\mu$ from $T$ to $S(L)$ such that $\varphi_{L} = \mu \circ h$.  Similarly, by Lemma \ref{lemma:h-star}, there is a surjective morphism $\mu^\star$ from $T$ to $S(L)^{\op}$ such that $\varphi_{L}^\da = \mu^\star \circ h$. 
	
	Let $(\mu, \mu^\star)$  be the product map from $T$ to $\scrS(L)$, defined as, for $t \in T$,
	\begin{align}
		(\mu, \mu^*) &\colon t \mapsto (\mu(t), \mu^{\star}(t))~.
	\end{align}
	Since $h$ is surjective this implies that there is a surjective $\star$-morphism namely 
	$(\mu, \mu^*)$ from $\scrT$ to $\scrS(L)$. 
\end{proof}

\subsection{Involution Varieties and Pseudovarieties} 

A \emph{pseudovariety of finite semigroups} is a  set of finite semigroups that is closed under finite direct products, subsemigroups and quotients. 

If $\scrS=(S, \star)$ and $\scrT=(T,\dagger)$ are two $\star$-semigroups then the {\em direct product} of $\scrS$ and $\scrT$, denoted $\scrS \times \scrT$, is the semigroup $S \times T$ with the involution $(x,y) \rightarrow (x^\star, y^\dagger)$. 

\begin{definition}
 A class of finite involution semigroups $\bfV$ is an {\em involution pseudovariety} if it is closed under taking finite direct products,  sub-$\star$-semigroups and quotients. 
\end{definition}

If $\bfS$ is a pseudovariety of finite semigroups, then 
$\bfS^\star$ is the pseudovariety of finite involution semigroups  whose semigroup reducts are in $\bfS$.

\begin{example}
\begin{enumerate}
\item $\LI$ is the pseudovariety of all finite semigroups $S$ such that $S$ is locally trivial, i.e., for each idempotent $e \in S$ and element $x\in S$, $exe=e$. Hence, $\LI^{\star}$ is the set of all involution semigroups whose semigroup reducts are in $\LI$.

\item $\Acom^\star$ is the pseudovariety of all finite involution semigroups $(S, \star)$ such that $S$ is commutative and aperiodic. Since $S$ is commutative, the identity function $\id$ is always an involution on $S$.   
\end{enumerate}
\end{example}

We let $\Reg_{A}$ denote the family of all regular languages over the alphabet $A$.
A \emph{variety of recognisable languages} is a function $\calV$ that maps each finite alphabet $A$ to a subset of $\Reg_{A}$ such that the following conditions hold:

\begin{enumerate}
    \item For each alphabet $A$, the family of languages $\calV(A)$ is closed under union, intersection and complement, i.e., $\calV(A)$ is a Boolean algebra.     
    \item For each $L \in \calV(A)$ and each letter $a\in A$, the quotients $a^{-1}L$ and $La^{-1}$ are in $\calV(A)$.
    \item For each morphism $h\colon A^+ \rightarrow  B^+$, if $L \in \calV(B)$, then $h^{-1}(L) \in \calV(A)$.
      \end{enumerate}

\begin{example}
\begin{enumerate}
\item $\calLI$ is the variety of languages where $\calLI(A)$ is the set of all languages $L \subseteq A^+$ that are  finite boolean combinations of languages of the form $uA^*$ or $A^*u$, for some $u \in A^+$. The language variety $\calLI$ corresponds to 
the pseudovariety $\LI$ \cite{PIN}.

\item Let $L(a,k) \subseteq A^{+}$ be the set of all words in which letter $a \in A$ appears exactly $k$ times.
$\calAcom$ is the variety where $\calAcom(A)$ is the set of all languages $L \subseteq A^+$ that are finite boolean combinations of languages of the form $L(a,k)$ \cite{PIN,StraubingBook}.

\end{enumerate}  
\end{example}

\begin{definition}
A \emph{involution variety of languages}, also called a $\star$-variety of languages, is a function $\calV$ that maps each involutory alphabet $\scrA=(A, \da)$ to a family of recognisable languages over the alphabet $A$ such that the following holds.
\begin{enumerate}
    \item For each $\scrA$, the set $\calV(\scrA)$ is a Boolean algebra.  
    \item For each $L \in \calV(\scrA)$ and each letter $a\in A$,  the quotients $a^{-1}L$ and $La^{-1}$ are in $\calV(\scrA)$.
    \item Let $h \colon (A^+, \da) \rightarrow (B^+, \star)$ be a morphism. Then, if $L \in \calV(\scrB)$, then $h^{-1}(L) \in \calV(\scrA)$.
    \item If $L \in \calV(\scrA)$, then $L^{\da}$ is in $\calV(\scrA)$.
      \end{enumerate}
\end{definition}

\subsection{Eilenberg-type correspondence} 

We define a correspondence between $\star$-varieties of languages and pseudovarieties of finite involution semigroups. We write $\bfV, \bfW,\dots$ for pseudovarieties of finite involution semigroups and $\calV, \calW, \dots$ for involution varieties of languages. With each involution variety of languages $\calV$, we associate a pseudovariety of involution semigroups $\bfV$, denoted as  $\calV \mapsto  \bfV$, where $\bfV$  is the pseudovariety generated by the syntactic $\star$-semigroups of languages in $\calV$. Likewise, we associate an involution variety of languages $\calV$ with every pseudovariety of involution semigroups $\bfV$, denoted as $\bfV \mapsto \calV$, where $\calV$ maps each involutory finite alphabet $(A, \da)$ to all the languages over $A$ recognised by $\star$-semigroups in $\bfV$.

\begin{theorem}
\label{thm:eilen}
The correspondences $\bfV \mapsto \calV$ and $\calV \mapsto \bfV $ define mutually inverse bijective correspondences between pseudovarieties
of involution semigroups and $\star$-varieties of languages.
\end{theorem}

Rest of this subsection is devoted to the proof of the above theorem. The following lemma is easy to verify.
    
\begin{lemma}The correspondences are well defined. i.e. 
\begin{enumerate}
\item If $\calV$ is a $\star$-variety of languages and $\calV \mapsto \bfV$, then $\bfV$ is a pseudovariety of involution semigroups.
\item If  $\bfV$ is a pseudovariety of involution semigroups and $\bfV \mapsto \calV$, then $\calV$ is a $\star$-variety of languages.
\end{enumerate}
\end{lemma}

Next we show that the maps are mutually inverse bijections. The proof follows a standard recipe; the one outlined here closely follows the proof of the Variety theorem for finite semigroups given in \cite{PIN}.

\begin{proposition} 
\label{prop:div-prod-syn}
Every involution semigroup $\scrS$ divides a finite direct product of syntactic $\star$-semigroups of languages recognised by $\scrS$.
\end{proposition}
\begin{proof}
Let $\scrS = (S, \star)$ be an involution semigroup. Let $A \subseteq S$ be such that $A$ generates $S$ and $\da$ be the involution $\star$ restricted to $A$. Let $\scrA = (A, \da)$.  Let $h\colon\scrA^+ \rightarrow \scrS$ be the surjective $\star$-morphism that maps each sequence of elements to their product in $S$. For $s \in S$, let $L_{s}$ denote the language recognised by the set $\{s\}$, in other words $L_s = h^{-1}(s)$.  Let $\scrS_s = (S_{s}, \star_{s})$ be the syntactic $\star$-semigroup of $L_s$ and let $\chi_{s}$ be the syntactic $\star$-morphism. 

Since $\scrS$ recognises the languages $L_{s}$, there exists a surjective morphism $\psi_s$ from $\scrS$ to $\scrS_s$, by Proposition \ref{prop:syn-prop}. Let $\Pi_{s\in S} \scrS_s$ denote the direct product of the syntactic $\star$-semigroups $\scrS_s$ corresponding to each element of $\scrS$. Let $\prod_{s \in S} \chi_s $ denote the product map  $u \mapsto \(\chi_s(u)\)_{s\in S}$, for $u\in A^{+}$.  Similarly, $\psi=\prod_{s \in S} \psi_s$  denotes the product map from $\scrS$ to $\prod_{s\in S} \scrS_s$. As product of involutory morphisms these maps are involutory morphisms. 

The various maps and semigroups are detailed in Figure \ref{fig:div-prod-syn}.

We claim that $\psi$ is injective. Assume $\psi(s) = \psi(s')$ for some $s,s' \in S$. Then in particular
$\psi_s(s)=\psi_s(s')$. This implies that for all $x,y \in S^1$, $xsy = s \Leftrightarrow xs'y = s$. In particular when $x=y=1$, this implies $s =s'$. Hence, $\psi$ is injective. Therefore $\scrS$ is isomorphic to a sub-$\star$-semigroup of $\prod_{s\in S} \scrS_s$.
\end{proof}

\begin{lemma}  The map $\bfV \mapsto \calV$  is injective.
\label{lemma:variety-injective}
\end{lemma}
\begin{proof}

Assume that $\bfV \mapsto \calV$ and $\bfW \mapsto \calW$. It suffices to show that if $\calV \subseteq \calW$, then $\bfV \subseteq \bfW$. If so, then 
\begin{align*}
\calV = \calW &\implies \calV \subseteq \calW \text{ and } \calW \subseteq \calV \\
&\implies  \bfV \subseteq \bfW \text{ and } \bfW \subseteq \bfV\\
&\implies \bfV = \bfW~,
\end{align*}
and the map is injective. Next, we prove this.

Assume that $\calV \subseteq \calW$. Let $\scrS = (S, \star)$ be a $\star$-semigroup in $\bfV$. We show that $\scrS$ is in $\bfW$ as well. Let $K$ be the set of all languages recognised by $\scrS$. The languages in $K$ are in $\calV$ and hence in $\calW$ also. This means there exist $\star$-semigroups recognising languages in $K$ in the pseudovariety $\bfW$, and hence the syntactic $\star$-semigroups of languages in $K$ are in  
 $\bfW$, by Proposition \ref{prop:syn-prop}. Therefore, their product is also in $\bfW$. Since $\scrS$ divides this product by Proposition \ref{prop:div-prod-syn}, $\scrS$ is in $\bfW$. Since $\scrS$ was arbitrary, this shows that $\bfV \subseteq \bfW$.
\end{proof}

\begin{lemma}
If $\calV \mapsto \bfV \mapsto \calW$ then $\calV = \calW$.
 \label{lemma:variety-main}
\end{lemma}
 \begin{proof}
Assume that $\calV \mapsto \bfV \mapsto \calW$.
Firstly, if $L \in \calV(\scrA)$, for some $\scrA=(A, \da)$, then one has $\S(L) \in \bfV$ by definition of the map $\calV \mapsto \bfV$. Hence,
 $L \in \calW(\scrA)$, by definition of the map $\bfV \mapsto \calW$.
Therefore,  $\calV(\scrA) \subseteq \cal{W}(\scrA)$, for every involutory alphabet $\scrA$.

Next we prove the inclusion $\calW(\scrA)\subseteq \calV(\scrA)$. Let $L \in \calW(\scrA)$.
Then, $\S(L)\in\bfV$, and since $\bfV$
is the pseudovariety generated by syntactic $\star$-semigroups of involutory languages in $\calV$, there
exists finitely many of languages $L_{1} \in \mathcal{V}(\scrA_{1}), \ldots, L_{n} \in \mathcal{V}(\scrA_{n})$, where $n>0$ and $\scrA_i = (A_i, \da_i)$,
such that $\scrS(L)$ divides $\S(L_{1})\times\dots\times \S(L_{n})$.
Let $\scrS = \S(L_{1})\times\dots\times \S(L_{n})$. Since $\S(L)$
divides $\scrS$, it is a quotient of a sub-$\star$-semigroup $\scrT = (T, \star)$ of $\scrS$. 

Observe
that $\scrT$ recognises $L$. Therefore there exists a surjective
$\star$-morphism $\varphi\colon\scrA^{+}\rightarrow \scrT$ and a subset
$P$ of $T$ such that $L = \varphi^{-1}(P)$. Let $\pi_i\colon\scrS\rightarrow \S(L_{i})$ be the $i^{th}$ projection defined by $\pi_{i}(s_{1},\ldots,s_{n})=s_{i}$.
Put $\varphi_{i}=\pi_{i}\circ\varphi$ and let $\chi_i\colon\scrA^+_{i} \rightarrow \scrS(L_{i})$
be the syntactic morphism of $L_{i}$. Since $\chi_{i}$ is surjective,
there exists a morphism of semigroups $\psi_{i}\colon\scrA^+\rightarrow \scrA_i^{+}$
such that $\varphi_{i}=\chi_{i}\circ\psi_{i}$. We can summarise the
situation in a diagram in Figure \ref{fig:variety-main}. 

\begin{figure}
\subfloat[][Various maps in the proof of Proposition \ref{prop:div-prod-syn}.]{\hfil
\xymatrixrowsep{.75in}
\xymatrixcolsep{.75in}
\SelectTips{cm}{}
\xymatrix{ 
{(A^+, \da)} \ar[d]_{\prod_{s\in S} \chi_s} \ar[r]^h & {\scrS} \ar[dl]^{\psi=\prod_{s \in S}\psi_s}\\
{\prod_{s\in S} \scrS_s}  & \\
}
}
\label{fig:div-prod-syn}
\centering
\subfloat[][Diagram in Lemma \ref{lemma:variety-main}]{\leavevmode
\xymatrixrowsep{.75in}
\xymatrixcolsep{.75in}
\SelectTips{cm}{}
\xymatrix{ \scrA^+ \ar@{->>}[d]^{\varphi} \ar[r]^{\psi_{i}} \ar[rd]^{\varphi_{i}} & \scrA^+_{i} \ar@{->>}[d]^{\chi_{i}} \\
\scrT\subseteq \scrS \ar[r]^-{\pi_{i}} &  \scrS(L_i)
}
}
\caption{}
\label{fig:variety-main}
\end{figure}

%
%

Next we show that $L\in \calV(\scrA)$ that proves the claim $\calW(\scrA)\subseteq \calV(\scrA)$.
Firstly,
\[
\hfil L=\varphi^{-1}(P) =\bigcup_{s \in P}\(  \varphi^{-1}(s)\)
\]

Since $\mathcal{V}(\scrA)$ is closed under union, it suffices
to establish that for every $s\in P$, one has that $\varphi^{-1}(s)$
belongs to $\calV(\scrA)$. Let $s=(s_{1},\ldots,s_{n})$,
then $s=\bigcap_{1\leq i\leq n}\pi_{i}^{-1}(s_{i})$. Hence
\[
\varphi^{-1}(s) = \bigcap_{1\leq i\leq n} \( \varphi^{-1}\(\pi_i^{-1}\(s_i\)\) \)=\bigcap_{1\leq i\leq n} \varphi_i^{-1}\(s_{i}\)~.
\]

As $\calV(\scrA)$ is closed under union and intersection, it is enough to show that
$ \varphi_i^{-1}\big(s_i\big)$ is in  $\cal{V}(\scrA)$ for each $1\leq i\leq n$.
Since $\varphi_{i}=\chi_{i}\circ\psi_{i}$, one has $\varphi_{i}^{-1}(s_{i})=\psi_{i}^{-1}(\chi_{i}^{-1}(s_{i}))$.
Now since $\mathcal{V}$ is a involution variety of languages, it suffices
to prove that $\chi_{i}^{-1}\big(s_{i}\big)$ is in $\calV(\scrA_i)$, that results from the  Lemma \ref{lemma:variety-last}.

 \end{proof} 
 
 \begin{lemma}
\label{lemma:variety-last}
Let $\mathcal{V}$ be an involution variety of  languages and let $\chi\colon\scrA^+\rightarrow \scrS(L)$
be the syntactic $\star$-morphism of a  language $L$ of $\calV(\scrA)$.
Then for every $x\in \scrS(L),\, \chi^{-1}(x)$ belongs to $\calV(\scrA)$.
\end{lemma}

\begin{proof}
Assume $\scrS(L) = (S, \star)$. Fix an element $x \in S$. Let $P=\chi(L)$. Then $L=\chi^{-1}(P)$. Setting $E=\{(s,\,t)\in S^{2}\mid sxt\in P\}$,
we claim that 
\[
 \{x\}=\bigcap_{(s,t) \in E} s^{-1}Pt^{-1} \setminus \bigcup_{(s,\,t) \in E^c }s^{-1}Pt^{-1}
\]
Let $R$ be the right-hand side of the above equation. It is clear
that $x$ belongs to $R$. Conversely, let $u$ be an element of $R$.
That means if $sxt\in P$, then $sut\in P$ and if $sxt\not \in P$, then $sut \not \in P$. It follows that
$u\sim_{L}x$ and thus $u=x$, which proves the claim.

%
%

Since $\chi^{-1}$
commutes with Boolean operations and quotients, $\chi^{-1}(x)$ is

\begin{align*}
 \chi^{-1}(x) & =\bigcap_{(s,t) \in E}  \chi^{-1} \( s^{-1}Pt^{-1}\)  
  \setminus \bigcup_{(s,\,t) \in E^c }  \chi^{-1} \( s^{-1}Pt^{-1}\) 
\end{align*}

which is a Boolean combination of quotients of $L$ and hence belongs to
$\calV(\scrA^+)$.
\end{proof}

 From Lemma \ref{lemma:variety-injective} and Lemma \ref{lemma:variety-main}
 we conclude Theorem \ref{thm:eilen}.

\subsection{Bilateral Semidirect Products of Involution Pseudovarieties} 
 
If $\bfV$ and $\bfW$ are pseudovarieties of finite semigroups, then the {\em bilateral semidirect product of the pseudovarieties}  
 is the pseudovariety generated by  all the bilateral semidirect products $S \sdp T$  with $S \in \bfV$, $T \in \bfW$.
Likewise, if $\bfV$ and $\bfW$ are pseudovarieties of involution semigroups, then
their {\em bilateral semidirect product} is the pseudovariety generated by all involutory semidirect products of the form $\scrS \sdp \scrT$, where $\scrS \in \bfV$ and  $\scrT \in \bfW$.

\begin{definition}

The {\em locally-hermitian semidirect product} of the pseudovarieties of involution semigroups $\bfV$ and $\bfW$ is the pseudovariety  generated by all locally hermitian products $\scrS \sdp \scrT$, where $\scrS \in \bfV$ and  $\scrT \in \bfW$.
\end{definition}

%

%
%
%
%
%
%
%
%
%

From Proposition \ref{prop:ltt-char}, we deduce
 
\begin{proposition}
The involution variety $\ltt$ corresponds to the bilateral semidirect product of the pseudovarieties $\Acom^*$  and $\LI^\star$.
\end{proposition}

Definition \ref{defn:approxr} and Theorem \ref{lemma:lrtt-acomli} can be generalized to the case of arbitrary involutory alphabets. We call the resulting involution variety of languages $\callrtt$.
Theorem \ref{lemma:lrtt-acomli} yields

\begin{theorem}
\label{thm:adj-char}
The involution variety $\callrtt$ corresponds to the locally hermitian semidirect product of the pseudovarieties $\Acom^*$  and $\LI^\star$.
\end{theorem}




\section{Conclusion}
\label{Sec:5} 

The membership problem for a class of languages $\calC$ asks: is the given regular language a member of $\calC$. The membership question is decidable for many families, the archetypical example is the class of aperiodic languages \cite{Schutzenberger}. The membership problem for locally threshold testable (\ltt) languages is equivalent to the problem of finding out if a given semigroup divides a semigroup in $\Acom * \LI$. Since $\bfV * \LI = \bfV *\bfD$ \cite{StraubingVD} (where $\bfD$ is the pseudovariety of {\em definite} semigroups that satisfies $se = e$ for all $e^2=e$) the problem is slightly simplified.
There are two main approaches to solving this problem, the first one is using the Derived 
Category theorem \cite{Tilson,TherienLTT,PinSeq, StraubingBook} and the other is using Presburger arithmetic \cite{mikolajLTT, PlaceLTT}. Both these approaches crucially depend on the fact that a semigroup $S$ divides a semigroup in  $\bfV * \bfD$ if and only if it divides a semigroup in $\bfV * \bfD_n$ (pseudovariety $\bfD_n$ satisfies the identity $yx_1\cdots x_n = x_1 \cdots x_n$),  where $n = |S|$ \cite{StraubingVD}. This is known as the delay theorem.

The key idea behind the delay theorem is that given a finite semigroup $S$, one can insert an idempotent in any product $x_1 \cdots x_n$, $x_i \in S$, of length $n=|S|$ without affecting the value of the product. This construction is not a left-to-right symmetric one. Hence, it is not clear if a delay theorem holds in the case of involutory semidirect products. This is the major hurdle in obtaining the decidability of membership of $\lrtt$.

The algebraic characterisation of the class $\lrtt$ presented in this work, is not decidable, i.e., the characterisation does not yield an algorithm to check if a given recognisable language is \lrtt or not. In the case of \ltt languages, the decidability is obtained by an algebraic characterisation in terms of categories \cite{Tilson}. The next logical step in our work is to extend the theory to categories.




\nocite{*}
\begin{thebibliography}{10}

\bibitem{BeauquierPin}
Dani{\`{e}}le Beauquier and Jean{-}{\'{E}}ric Pin.
\newblock Languages and scanners.
\newblock {\em Theoretical Computer Science}, 84(1):3--21, 1991.

\bibitem{mikolajLTT}
Miko{\l}aj Boja{\'{n}}czyk.
\newblock A new algorithm for testing if a regular language is locally
  threshold testable.
\newblock {\em Information Processing Letters}, 104(3):91 -- 94, 2007.
\newblock URL:
  \url{http://www.sciencedirect.com/science/article/pii/S002001900700141X},
  \href {https://doi.org/https://doi.org/10.1016/j.ipl.2007.05.015}
  {\path{doi:https://doi.org/10.1016/j.ipl.2007.05.015}}.

\bibitem{mikolajmonad}
Miko{\l}aj Boja{\'{n}}czyk.
\newblock Recognisable languages over monads.
\newblock In Igor Potapov, editor, {\em Developments in Language Theory}, pages
  1--13, Cham, 2015. Springer International Publishing.

\bibitem{mikolajmonadarxiv}
Miko{\l}aj Boja{\'{n}}czyk.
\newblock Recognisable languages over monads.
\newblock {\em CoRR}, abs/1502.04898, 2015.
\newblock URL: \url{http://arxiv.org/abs/1502.04898}, \href
  {http://arxiv.org/abs/1502.04898} {\path{arXiv:1502.04898}}.

\bibitem{BrozozoskiSimon}
Janusz~A. Brzozowski and Imre Simon.
\newblock Characterizations of locally testable events.
\newblock In {\em Proceedings of the 12th Annual Symposium on Switching and
  Automata Theory (Swat 1971)}, SWAT '71, pages 166--176, 1971.

\bibitem{Dolinka}
Sini\v{s}a Crvenkovi\'{c} and Igor Dolinka.
\newblock Varieties of involution semigroups and involution semirings: a
  survey.
\newblock In {\em Proceedings of the International Conference ``Contemporary
  Developments in Mathematics" (Banja Luka, 2000)}, Bulletin of Society of
  Mathematicians of Banja Luka, pages 7--47, 2000.

\bibitem{ebbinghaus2005finite}
Heinz-Dieter Ebbinghaus and J{\"o}rg Flum.
\newblock {\em Finite model theory}.
\newblock Springer Science \& Business Media, 2005.

\bibitem{PAG}
Paul Gastin, Amaldev Manuel, and R.~Govind.
\newblock Logics for reversible regular languages and semigroups with
  involution.
\newblock In {\em Developments in Language Theory}, pages 182--191, 2019.

\bibitem{PAG-Arxiv}
Paul Gastin, Amaldev Manuel, and R.~Govind.
\newblock Logics for reversible regular languages and semigroups with
  involution.
\newblock {\em CoRR}, abs/1907.01214, 2019.
\newblock URL: \url{http://arxiv.org/abs/1907.01214}, \href
  {http://arxiv.org/abs/1907.01214} {\path{arXiv:1907.01214}}.

\bibitem{holcombe}
Mike Holcombe and William Michael~Lloyd Holcombe.
\newblock {\em Algebraic automata theory}, volume~1.
\newblock Cambridge University Press, 2004.

\bibitem{Straubing1}
Andreas Krebs, Kamal Lodaya, Paritosh Pandya, and Howard Straubing.
\newblock Two-variable logic with a between relation.
\newblock In {\em Proceedings of the 31st Annual ACM/IEEE Symposium on Logic in
  Computer Science}, LICS '16, pages 106--115, 2016.

\bibitem{Straubing2}
Andreas Krebs, Kamal Lodaya, Paritosh~K. Pandya, and Howard Straubing.
\newblock {An Algebraic Decision Procedure for Two-Variable Logic with a
  Between Relation}.
\newblock In {\em 27th EACSL Annual Conference on Computer Science Logic (CSL
  2018)}, volume 119 of {\em Leibniz International Proceedings in Informatics
  (LIPIcs)}, pages 28:1--28:17, 2018.

\bibitem{Straubing3}
Andreas Krebs, Kamal Lodaya, Paritosh~K Pandya, and Howard Straubing.
\newblock Two-variable logics with some betweenness relations: Expressiveness,
  satisfiability and membership.
\newblock {\em arXiv preprint arXiv:1902.05905}, 2019.

\bibitem{freesem}
Mark~V Lawson.
\newblock {\em Inverse semigroups: the theory of partial symmetries}.
\newblock World Scientific, 1998.

\bibitem{McNaughton}
Robert McNaughton and Seymour~A. Papert.
\newblock {\em Counter-Free Automata (M.I.T. Research Monograph No. 65)}.
\newblock The MIT Press, 1971.

\bibitem{PIN}
Jean{-}{\'{E}}ric Pin.
\newblock {\em Mathematical Foundations of Automata Theory}.
\newblock URL: \url{https://www.irif.fr/~jep/PDF/MPRI/MPRI.pdf}.

\bibitem{pin1986varieties}
Jean{-}{\'{E}}ric Pin.
\newblock {\em Varieties of Formal Languages}.
\newblock Foundations of computer science. North Oxford Academic, 1986.

\bibitem{PinSeq}
Jean{-}{\'{E}}ric Pin.
\newblock Expressive power of existential first-order sentences of
  {B}\"{u}chi's sequential calculus.
\newblock {\em Discrete Math.}, 291(1–3):155–174, March 2005.
\newblock \href {https://doi.org/10.1016/j.disc.2004.04.027}
  {\path{doi:10.1016/j.disc.2004.04.027}}.

\bibitem{PlaceLTT}
Thomas Place, Lorijn van Rooijen, and Marc Zeitoun.
\newblock On separation by locally testable and locally threshold testable
  languages.
\newblock {\em Logical Methods in Computer Science}, 10(3), Sep 2014.
\newblock URL: \url{http://dx.doi.org/10.2168/LMCS-10(3:24)2014}, \href
  {https://doi.org/10.2168/lmcs-10(3:24)2014}
  {\path{doi:10.2168/lmcs-10(3:24)2014}}.

\bibitem{Schutzenberger}
Marcel-Paul Sch{\"{u}}tzenberger.
\newblock On finite monoids having only trivial subgroups.
\newblock {\em Information and Control}, 8(2):190--194, 1965.

\bibitem{StraubingVD}
Howard Straubing.
\newblock Finite semigroup varieties of the form {$\mathrm{\bf V\star D}$}.
\newblock {\em Journal of Pure and Applied Algebra}, 36:53 -- 94, 1985.

\bibitem{StraubingBook}
Howard Straubing.
\newblock {\em Finite Automata, Formal Logic, and Circuit Complexity}.
\newblock Birkh\"auser Verlag, Basel, Switzerland, 1994.

\bibitem{WolfgangThomas}
Wolfgang Thomas.
\newblock Classifying regular events in symbolic logic.
\newblock {\em Journal of Computer and System Sciences}, 25(3):360--376, 1982.

\bibitem{TherienLTT}
Denis Thérien and Alex Weiss.
\newblock Graph congruences and wreath products.
\newblock {\em Journal of Pure and Applied Algebra}, 36:205 -- 215, 1985.
\newblock URL:
  \url{http://www.sciencedirect.com/science/article/pii/0022404985900714},
  \href {https://doi.org/https://doi.org/10.1016/0022-4049(85)90071-4}
  {\path{doi:https://doi.org/10.1016/0022-4049(85)90071-4}}.

\bibitem{Tilson}
Bret Tilson.
\newblock Categories as algebra: An essential ingredient in the theory of
  monoids.
\newblock {\em Journal of Pure and Applied Algebra}, 48(1-2):83--198, 1987.

\bibitem{Zalcstein1972}
Yechezkel Zalcstein.
\newblock Locally testable semigroups.
\newblock {\em Semigroup Forum}, 5(1):216--227, Dec 1972.
\newblock \href {https://doi.org/10.1007/BF02572893}
  {\path{doi:10.1007/BF02572893}}.

\end{thebibliography}
\end{document}